\documentclass[envcountsect,envcountsame]{llncs}
\usepackage{graphicx}
\usepackage{tikz}
\usepackage[english]{babel}
\usepackage[utf8]{inputenc}
\usepackage{url}
\usepackage{booktabs}
\usepackage{xspace}
\usepackage{amsmath}
\usepackage{amssymb}
\usepackage{float}
\usepackage{braket}
\usepackage{xcolor}
\usepackage[colorlinks=true, linkcolor=linkcolor, urlcolor=urlcolor, citecolor=citecolor]{hyperref}

\definecolor{linkcolor}{rgb}{0.65,0,0}
\definecolor{citecolor}{rgb}{0,0.65,0}
\definecolor{urlcolor}{rgb}{0,0,0.65}
\def\subheading#1{\medskip\noindent{\boldmath\textbf{#1}}~\ignorespaces}
\spnewtheorem{construction}[theorem]{Construction}{\bfseries}{\itshape}
\providecommand{\mytitle}[2][***]{\title{#2}}
\providecommand{\myauthor}[1]{\author{#1}}
\providecommand{\myinstitute}[1]{\institute{\renewcommand{\inst}[1]{\!\!}#1}}
\newcommand{\eps}{\varepsilon}
\newcommand{\bin}{\ensuremath{\{0,1\}}}
\newcommand{\olrk}[1]{  \ifx\nursymbol#1\else\!\!\mskip4.5mu plus 0.5mu\left(#1\right)\fi}
\newcommand{\negl}[1]{  \ensuremath{\textrm{negl}\olrk{#1}}}
\newcommand{\poly}[1]{  \ensuremath{\operatorname{poly}\olrk{#1}}}

\newcommand{\xor}{\oplus}
\newcommand{\ppttxt}{probabilistic polynomial-time\xspace}
\newcommand{\NN}{\mathbb{N}}
\newcommand{\RR}{\mathbb{R}}

\newcommand{\exec}{\ensuremath{\longleftarrow}}
\newcommand{\rand}{\ensuremath{\stackrel{\$}{\longleftarrow}}}
\newcommand{\F}{\ensuremath{F}}

\newcommand{\ff}{\ensuremath{\mathcal{F}}}

\newcommand{\range}{\ensuremath{\mathcal{Y}}}
\newcommand{\kSpace}{\ensuremath{\mathcal{K}}}

\newcommand{\msgSpace}{\ensuremath{\mathcal{M}}}
\newcommand{\init}{\ensuremath{\mathcal{I}}\xspace}
\newcommand{\adver}{\ensuremath{\mathcal{A}}\xspace}
\newcommand{\D}{\ensuremath{\mathcal{D}}\xspace}
\newcommand{\A}{\adver}
\newcommand{\C}{\cdv}
\newcommand{\K}{\ensuremath{\mathcal{K}}\xspace}
\newcommand{\M}{\ensuremath{\mathcal{M}}\xspace}
\newcommand{\X}{\ensuremath{\mathcal{X}}\xspace}
\newcommand{\Y}{\ensuremath{\mathcal{Y}}\xspace}
\newcommand{\MA}{\ensuremath{\M^\A}\xspace}
\newcommand{\cdv}{\ensuremath{\mathcal{C}}\xspace}
\newcommand{\pr}{\ensuremath{\mathsf{Pr}}} 

\newcommand{\oracle}{\ensuremath{\mathcal{O}}\xspace}
\newcommand{\defas}{\stackrel{def}{=}}

\newcommand{\Gen}{\ensuremath{\mathsf{Gen}\xspace}}
\newcommand{\Enc}{\ensuremath{\mathsf{Enc}\xspace}}
\newcommand{\Dec}{\ensuremath{\mathsf{Dec}\xspace}}

\newcommand{\E}{\ensuremath{\mathcal{E}}}
\newcommand{\from}{\ensuremath{\leftarrow}}
\newcommand{\hilbert}{\ensuremath{\mathcal{H}\xspace}}
\renewcommand{\poly}[1]{\ensuremath{{\text{poly}\left(#1\right)}\xspace}}
\renewcommand{\sim}{\ensuremath{\mathcal{S}}\xspace}
\newcommand{\supp}[1]{\ensuremath{\text{Supp}\left(#1\right)\xspace}}

\renewcommand{\negl}[1]{\ensuremath{\text{negl}\left(#1\right)\xspace}}
\newcommand{\sketbra}[2]{{\ensuremath{\lvert #1\rangle\!\langle #2\rvert}}}
\newcommand{\lketbra}[2]{{\ensuremath{\left\lvert #1\right\rangle\!\!\left\langle #2\right\rvert}}}
\newcommand{\ketbra}[2]{\if@display\lketbra{#1}{#2}\else\sketbra{#1}{#2}\fi}

\newcommand{\proj}[1]{\ketbra{#1}{#1}}
\newcommand{\one}{\mathbb{I}}
\DeclareMathOperator{\tr}{tr}

\begin{document}

\mytitle{Semantic Security and Indistinguishability\\in the Quantum World}
\author{}
\institute{}
\subtitle{{\small \today\footnote{An extended abstract of this work appears in the proceedings of CRYPTO 2016. This is the full version.} \vspace{-0.2in}}}
\myauthor{Tommaso Gagliardoni\inst{1} \and Andreas H\"ulsing\inst{2} \and Christian Schaffner\inst{3,4,5}} 
\myinstitute{
\inst{1} CASED, Technische Universit\"at Darmstadt, Germany \\ {\tt tommaso@gagliardoni.net}\and 
\inst{2} TU Eindhoven, The Netherlands \\ {\tt andreas@huelsing.net}\and
\inst{3} Institute for Logic, Language and Compuation (ILLC),\\ University of Amsterdam, The Netherlands\\ {\tt c.schaffner@uva.nl} \and
\inst{4} Centrum Wiskunde \& Informatica (CWI) Amsterdam, The Netherlands \and
\inst{5} \href{http://www.qusoft.org}{QuSoft}, The Netherlands\\
}

\maketitle

\begin{abstract}
\vspace{-0.2in}
At CRYPTO 2013, Boneh and Zhandry initiated the study of quantum-secure encryption. They proposed first indistinguishability definitions for the quantum world where the actual indistinguishability only holds for classical messages, and they provide arguments why it might be hard to achieve a stronger notion. In this work, we show that stronger notions are achievable, where the indistinguishability holds for quantum superpositions of messages. We investigate exhaustively the possibilities and subtle differences in defining such a quantum indistinguishability notion for symmetric-key encryption schemes. We justify our stronger definition by showing its equivalence to novel quantum semantic-security notions that we introduce. Furthermore, we show that our new security definitions cannot be achieved by a large class of ciphers -- those which are quasi-preserving the message length. On the other hand, we provide a secure construction based on quantum-resistant pseudorandom permutations; this construction can be used as a generic transformation for turning a large class of encryption schemes into quantum indistinguishable and hence quantum semantically secure ones. Moreover, our construction is the first completely classical encryption scheme shown to be secure against an even stronger notion of indistinguishability, which was previously known to be achievable only by using quantum messages and arbitrary quantum encryption circuits.
\end{abstract}

\setcounter{page}{1}

\vspace{-0.3in}

\section{Introduction}

Quantum computers~\cite{NC00} threaten many cryptographic schemes. By using Shor's algorithm \cite{Sho94} and its variants \cite{Watrous01}, an adversary in possession of a quantum computer can break the security of every scheme based on factorization and discrete logarithms, including RSA, ElGamal, elliptic-curve primitives and many others. Moreover, longer keys and output lengths are required in order to maintain the security of block ciphers and hash functions \cite{Grover96,Brassard97}. These difficulties led to the development of {\em post-quantum cryptography}~\cite{BuchPQ}, i.e., classical cryptography resistant against quantum adversaries. 

When modeling the security of cryptographic schemes, care must be taken in defining exactly what property one wants to achieve. In classical security models, all parties and communications are classical. When these notions are used to prove {\em post-quantum} security, one must consider adversaries having access to a quantum computer. This means that, while the communication between the adversary and the user is still classical, the adversary might carry out computations on a quantum computer. 

Such post-quantum notions of security turn out to be unsatisfying in certain scenarios. For instance, consider quantum adversaries able to use {\em quantum superpositions} of messages $\sum_x \alpha_x \ket{x}$ instead of classical messages when communicating with the user, even though the cryptographic primitive is still classical. This kind of scenario is considered, e.g., in~\cite{Boneh2013,Damgard13,Unruh12,Watrous09,Zhandry2012}. Such a setting might for example occur in a situation where one party using a quantum computer encrypts messages for another party that uses a classical computer and an adversary is able to observe the outcome of the quantum computation before measurement. Other examples are an attacker which is able to trick a classical device into showing quantum behavior, or a classical scheme which is used as subprotocol in a larger quantum protocol. 
Another possibility occurs when using obfuscation. There are applications where one might want to distribute the obfuscated code of a symmetric-key encryption scheme (with the secret key hardcoded) in order to allow a third party to generate ciphertexts without being able to retrieve the key - think of this as building public-key encryption from symmetric-key encryption using Indistinguishability Obfuscation. Because in these cases an adversary receives the classical code for producing encryptions, he could implement the code on his local quantum computer and query the resulting quantum circuit on a superporition of inputs. Moreover, even in quantum reductions for classical schemes situations could arise where superposition access is needed. A typical example are impossibility results (such as meta-reductions~\cite{FSQROM13}), where giving the adversary additional power often rules out a broader range of secure reductions. Notions covering such settings are often called {\em quantum-security} notions. In this work we propose new quantum-security notions for encryption schemes. 

For encryption, the notion of {\em semantic security}~\cite{GM84,Goldreich2004} has been traditionally used. This notion models in abstract terms the fact that, without the corresponding decryption key, it is impossible not only to correctly decrypt a ciphertext, but even to recover any non-trivial information about the underlying plaintext. The exact definition of semantic security is cumbersome to work with in security proofs as it is simulation-based. Therefore, the simpler notion of {\em ciphertext indistinguishability} has been introduced. This notion is given in terms of an interactive game where an adversary has to distinguish the encryptions of two messages of his choice. The advantage of this definition is that it is easier to work with than (but equivalent to) semantic security.

To the best of our knowledge, no quantum semantic-security notions for classical encryption schemes have been proposed so far. For indistinguishability, Boneh and Zhandry introduced indistinguishability notions for quantum-secure encryption under chosen-plaintext attacks in a recent work~\cite{Boneh2013}. They consider a model (IND-qCPA) where a quantum adversary can query the encrypting device in superposition during a learning phase but is limited to classical communication during the actual challenge phase. However, in the symmetric-key scenario, this approach has the following shortcoming: If we assume that an adversary can get quantum access in a learning phase, it seems unreasonable to assume that he cannot get such access when the actual message of interest is encrypted. Boneh and Zhandry showed that a seemingly natural notion of quantum indistinguishability is unachievable. In order to restore a meaningful definition, they resorted to the compromise of IND-qCPA.

\subheading{Our contributions.} In this paper we achieve two main results. On the one hand, we initiate the study of semantic security in the quantum world, providing new definitions and a thorough discussion about the motivations and difficulties of modeling these notions correctly. This study is concluded by a suitable definition of {\em quantum semantic security under chosen plaintext attacks (qSEM-qCPA)}. On the other hand, we extend the fundamental work initiated in~\cite{Boneh2013} in finding suitable notions of indistinguishability in the quantum world. We show that the compromise that had to be reached there in order to define an achievable notion instead of a more natural one (i.e., IND-qCPA vs. fqIND-qCPA) can be overcome -- although not trivially. We show how various other possible notions of quantum indistinguishability can be defined. All these security notions span a tree of possibilities which we analyze exhaustively in order to find the most suitable definition of {\em quantum indistinguishability under chosen plaintext attacks (qIND-qCPA)}. We prove this notion to be achievable, strictly stronger than IND-qCPA, and equivalent to qSEM-qCPA, thereby completing an elegant framework of security notions in the quantum world, see Figure~\ref{fig:relations} below for an overview.

Furthermore, we formally define the notion of a {\em core function} and {\em quasi--length-preserving ciphers} -- encryption schemes which essentially do not increase the plaintext size, such as stream ciphers and many block ciphers including AES -- and we show the impossibility of achieving our new security notion for this kind of schemes. While this impossibility might look worrying from an application perspective, we also present a transformation that turns a block cipher into an encryption scheme fulfilling our notion. This transformation also works in respect to an even stronger notion of indistinguishability in the quantum world, which was introduced in~\cite{BJ15}, and previously only known to be achievable in the setting of {\em computational quantum encryption}, that is, the scenario where all the parties have quantum computing capabilities, and encryption is performed through arbitrary quantum circuits operating on quantum data. Even if this scenario goes in a very different direction from the scope of our work, it is interesting to note that our construction is the first fully classical scheme secure even in respect to such a purely quantum notion of security.

\subheading{The `frozen smart-card' example.} In order to clarify why quantum security allows the adversary {\em quantum superposition access} to classical primitives - as opposed to the case of post-quantum security - we give a motivating example. In this mind experiment, we consider a not-so-distant future where the target of an attack is a tiny encryption chip, e.g., integrated into an RFID tag or smart-card. It is reasonable to assume that it will include elements of technology currently researched but undeployed (i.e., extreme miniaturization, optical electronics, etc.) Regardless, the chip we consider is a purely classical device, performing classical encryption (e.g. AES) on classical inputs, and outputting classical outputs. 

Consider an adversary equipped with some future technology which subjects the device to a fault-injection environment, by varying the physical parameters (such as temperature, power, speed, etc.) under which the device usually operates. As a figurative example, our `quantum hacker' could place the chip into an isolation pod, which keeps the device at a very low temperature and shields it from any external electromagnetic or thermal interference. This situation would be analogous to what happens when security researchers perform side channel analysis on cryptographic hardware in nowaday's labs, using techniques such as thermal or electromagnetic manipulation which were previously considered futuristic.
There is no guarantee that, under these conditions, the chip does not start to show full or partial quantum behaviour. At this point, the adversary could query the device on a superposition of plaintexts by using, e.g., a laser and an array of beam splitters when feeding signals into the chip via optic fiber. 

It is unclear today what a future attacker might be able to achieve using such an attack. As traditionally done in cryptography, we assume the worst-case scenario where the attacker can actually query the target device in superposition. Classical security notions such as IND-CPA do not cover this scenario while our new notion qIND-qCPA does. This setting is an example of what we mean by `tricking classical parties into quantum behaviour'.

\subheading{Related work.} The idea of considering scenarios where a quantum adversary can force other parties into quantum behaviour has been first considered in~\cite{Damgard13}. Attacks exploiting classical encryptions in quantum superposition have been described in~\cite{KuwakadoM10,KuwakadoM12,KLLN16arxiv,SS16arxiv}. In~\cite{Boneh2013} the authors also consider the security of signature schemes where the adversary can have quantum access to a signing oracle. Quantum superposition queries have also been investigated relatively to the random oracle model~\cite{QROM}. Another quantum indistinguishability notion has been suggested (but not further analyzed) by Velema in~\cite{Velema2013}. Prior work has considered the security of quantum methods to encrypt classical data in the computational setting~\cite{sim2,sim3}. In concurrent and independent work, Broadbent and Jeffery~\cite{BJ15} introduce indistinguishability notions for the public- and secret-key encryption of quantum messages in the context of fully homomorphic quantum computation. We refer to Page~\pageref{BJcomparison} for a more detailed description of how their definitions relate to our framework. A more complete overview for these notions, including semantic security for quantum encryption schemes, can be found in another concurrent work~\cite{Alagic+16arxiv}.

\section{Preliminaries}\label{sec:preliminaries}
In this section, we briefly recall the classical security notions for encryption schemes secure against chosen plaintext attacks (CPA). In addition, we revisit the two existing indistinguishability notions for the quantum world. We start by introducing notation we will use throughout the paper.

We say that a function $f: \NN \to \RR$ is {\em polynomially bounded} iff there exists a polynomial $p$ and a value $\bar{n} \in \NN$ such that: for every $n\geq \bar{n}$ we have that $f(n) \leq p(n)$; in this case we will just write $f = \poly{n}$. We say that a function $\eps: \NN \to \RR$ is {\em negligible}, if and only if for every polynomial $p$, there exists an $n_p \in \NN$ such that $\eps(n) \leq \frac{1}{p(n)}$ for every $n \geq n_p$; in this case we will just write $\eps = \negl{n}$. In this work, we focus on secret-key encryption schemes.
In all that follows we use $n\in\NN$ as the security parameter.

\begin{definition}[Secret-key encryption scheme \cite{Goldreich2004}] 
A {\em secret-key encryption scheme} is a triple of probabilistic polynomial-time algorithms $(\Gen$, $\Enc$, $\Dec)$ operating on a message space $\msgSpace= \bin^m$ (where $m=\poly{n} \in \NN$) that fulfills the following two conditions: 
\begin{enumerate}
\item The key generation algorithm $\Gen(1^n)$ on input of security parameter $n$ in unary outputs a bitstring $k$. 
\item \label{ENC:correctness}For all $k$ in the range of $\Gen(1^n)$ and any message $x \in \msgSpace$, the algorithms $\Enc$ (encryption) and $\Dec$ (decryption) satisfy 
$\pr [\Dec(k,\Enc(k,x)) = x] = 1$, where the probability is taken over the internal coin tosses of $\Enc$ and $\Dec$.
\end{enumerate}
\end{definition}

We write $\kSpace$ for the range of $\Gen(1^n)$ (the key space) and $\Enc_k(x)$~for~$\Enc(k,x)$.

\subsection{Classical Security Notions: IND-CPA and SEM-CPA.}

We turn to security notions for encryption schemes. In this work, we will only look at the notions of indistinguishability of ciphertexts under adaptively chosen plaintext attack (IND-CPA), and semantic security under adaptively chosen plaintext attack (SEM-CPA), which are known to be equivalent (e.g., \cite{Goldreich2004}).

\subheading{Game-based definitions.} In general these notions can be defined as a game between a challenger \C and an adversary \A. First, \C generates a legitimate key running $k\exec \Gen(1^n)$ which he uses throughout the game. The game starts with a first learning phase. A challenge phase follows where \A receives a challenge. Afterwards, a second learning phase follows, and finally \A has to output a solution. The learning phases define the type of attack, and the challenge phase the notion captured by the game. We give all our definitions by referring to this game framework and by defining a learning and a challenge phase.

\subheading{The CPA learning phase:} \A is allowed to adaptively ask \C for encryptions of messages of his choice. \C answers the queries using key $k$. Note that this is equivalent to saying that \A gets oracle access to an encryption oracle that was initialized with key $k$.

\subheading{The IND challenge phase:} \A defines a challenge template consisting of two 
equal-length messages $x_0,x_1$, and sends it to \C. The challenger \C samples a random bit $b \rand \bin$ uniformly at random, and replies with the encryption $\Enc_k(x_b)$. \A's goal is to guess $b$.

\begin{definition}[IND-CPA]\label{def:indcpa}
A secret-key encryption scheme is said to be {\em IND-CPA secure} if the success probability of any \ppttxt adversary winning the game defined by CPA learning phases and an IND challenge phase is at most negligibly (in $n$) close to $1/2$.
\end{definition}

\subheading{The SEM challenge phase:} \A sends \C a challenge template $(S_m,h_m,f_m)$ consisting of a poly-sized circuit $S_m$ specifying a distribution over $m$-bit long plaintexts, an advise function $h_m:\bin^m\rightarrow\bin^*$, and a target function $f_m:\bin^m\rightarrow\bin^*$. The challenger \C replies with the pair $(\Enc_k(x), h_m(x))$ where $x$ is sampled according to $S_m$. \A's challenge is to output $f_m(x)$.

In the definition of semantic security it is not required that \A's probability of winning the game is always negligible. Instead, \A's success probability is compared to that of a simulator \sim that plays in a \emph{reduced game}: On one hand, \sim gets no learning phases. On the other hand, during the challenge phase, \sim does not receive the ciphertext but only the output of the advice function. This use of a simulator is what makes the notion hard to work with in proofs as one has to construct a simulator for every possible \A to prove a scheme secure.

\begin{definition}[SEM-CPA]\label{def:semcpa}
A secret-key encryption scheme is said to be {\em SEM-CPA secure} if for any \ppttxt adversary \A there exists a \ppttxt simulator \sim such that the challenge templates produced by \sim and \A are identically distributed and the success probability of \A winning the game defined by CPA learning phases and a SEM challenge phase (computed over the coins of \A, \Gen , and $S_m$) is negligibly close (in $n$) to the success probability of \sim winning the reduced game.
\end{definition}

Semantic security models what we want an encryption scheme to achieve: An adversary given a ciphertext can learn nothing about the encrypted message which he could not also learn from his knowledge of the message distribution and possibly existing side-information (modeled by $h_m$). Indistinguishability of ciphertexts is an equivalent technical notion introduced to simplify proofs.

\subsection{Previous Notions of Security in the Quantum World}

We briefly recall the results from \cite{Boneh2013} about quantum indistinguishability notions. We refer to \cite{NC00} for commonly used notation and quantum information-theoretic concepts. Given security parameter $n$, let $\{\hilbert_n\}_n$ be a family of complex Hilbert spaces such that $\dim{\hilbert_n} = 2^{\poly{n}}$. We assume that $\hilbert_n$ contains all the subspaces where the message states, the ciphertext states and any auxiliary state live. For the sake of simplicity we will not make a distinction when writing that a state $\ket{\varphi}$ belongs to one particular subspace, and we will omit the index $n$ when the security parameter is implicit, therefore writing just $\ket{\varphi} \in \hilbert$. We will denote pure states with ket notation, e.g., $\ket{\varphi}$, while mixed states will be denoted by lowercase Greek letters, e.g. $\rho$. We start by defining what we call a {\em classical description of a quantum state}:

\begin{definition}[Classical Description]\label{def:classrepr}
A {\em classical description} of a quantum state $\rho$ is a (classical) bitstring describing a quantum circuit $S$ which (takes no input but starts from a fixed initial state $\ket{0}$ and) outputs $\rho$. 
\end{definition}

This definition will be used later in our new notions of security. We deviate here from the traditional meaning of `classical description' referring to individual numerical entries of the density matrix. The reason is that our definition also covers the cases where those numerical entries are not easily computable, as long as we can give an explicit constructive procedure for that state. Clearly, every pure quantum state $\ket{\varphi}$ has a classical description (given by a description of the quantum circuit which implements the unitary that maps $\ket{0}$ to $\ket{\varphi}$. The classical description of a mixed state $\rho_A$ is given by the circuit which first creates a purification $\ket{\varphi}_{AR}$ of $\rho_A$ and then only outputs the $A$ register. Note that a state admitting a classical description cannot be entangled with any other system.

For encryption, following the approach in \cite{Boneh2013} and many other works, we define the following:

\begin{definition}[Quantum Encryption Oracle \cite{Boneh2013}]\label{def:quantumEncryptionOracle}
Let $\Enc$ be the encryption algorithm of a secret-key encryption scheme $\E$. We define the {\em quantum encryption oracle} $U_{\Enc_k}$ associated with $\E$ and initialized with key $k$ as (a family of) unitary operators defined by: 
\begin{equation}\label{eq:Uenc}
U_{\Enc_k}: \sum_{x,y} \alpha_{x,y} \ket{x}\ket{y} \mapsto \sum_{x,y} \alpha_{x,y} \ket{x}\ket{y \oplus \Enc_k(x)}
\end{equation}
where the same randomness $r$ is used in superposition in all the executions of $\Enc_k(x)$ within one query\footnote{As shown in~\cite{Boneh2013}, this is not restrictive.} -- for each new query, a fresh independent $r$ is used.
\end{definition}

The first indistinguishability notion proposed in~\cite{Boneh2013} replaces all classical communication between \A and \C by quantum communication. \A and \C are now quantum circuits operating on quantum states, and sharing a certain number of qubits (the quantum communication register). The definition for the new security game is obtained from Definition~\ref{def:indcpa} by changing the learning and challenge phases as follows:

\subheading{Quantum CPA learning phase (qCPA):} \A gets oracle access to $U_{\Enc_k}$.

\subheading{Fully quantum IND challenge phase (fqIND):} \A prepares the communication register in the state $\sum_{x_0,x_1,y} \alpha_{x_0,x_1,y} \ket{x_0}\ket{x_1}\ket{y}$, consisting of two $m$-qubit states (the two input-message superpositions) and an ancilla state to store the ciphertext. \C samples a bit $b\rand\bin$ and applies the transformation:
$$
\sum_{x_0,x_1,y} \alpha_{x_0,x_1,y} \ket{x_0}\ket{x_1}\ket{y} \mapsto \sum_{x_0,x_1,y} \alpha_{x_0,x_1,y} \ket{x_0}\ket{x_1}\ket{y \oplus \Enc_k(x_b)}. 
$$
\A's goal is to output $b$.

The resulting security notion in~\cite{Boneh2013} is called {\em indistinguishability under fully quantum chosen-message attacks (IND-fqCPA)}. We decided to rename it to {\em fully quantum indistinguishability under quantum chosen-message attacks (fqIND-qCPA)} in order to fit into our naming scheme: It consists of a quantum CPA learning phase and a fully quantum IND challenge phase.

\begin{definition}[fqIND-qCPA]
A secret-key encryption scheme is said to be {\em fqIND-qCPA secure} if the success probability of any quantum \ppttxt adversary winning the game defined by qCPA learning phases and a fqIND challenge phase is at most negligibly close (in $n$) to $1/2$.  
\end{definition}

As already observed in \cite{Boneh2013}, this notion is unachievable. The separation by Boneh and Zhandry exploits the entanglement of quantum states, namely the fact that entanglement can be created between plaintext and ciphertext.

\begin{theorem}[BZ attack {\cite[Theorem 4.2]{Boneh2013}}]\label{thm:BZatt}
No symmetric-key encryption scheme can achieve fqIND-qCPA security.
\end{theorem}

\begin{proof}
The attack works as follows: The adversary \A chooses as challenge messages the states $\ket{0^m}$ and $H\ket{0^m}$ (where $H$ denotes the $m$-fold tensor Hadamard transform), i.e.\ he prepares the register in the state $\sum_x \frac{1}{2^{m/2}} \ket{0^m,x,0^m}$. When the challenger \C performs the encryption, we can have two cases:
\begin{itemize}
\item if $b=0$, i.e. the first message state is chosen, the state is transformed into
$$
\sum_x \frac{1}{2^{m/2}} \ket{0^m,x,\Enc_k(0^m)} % = \sum_x \frac{1}{2^{m/2}} \ket{0^m} \otimes \ket{x} \otimes \ket{\Enc_k(0^m)} =
= \ket{0^m} \otimes H\ket{0^m} \otimes \ket{\Enc_k(0^m)};
$$
\item if $b=1$, i.e. the second message state is chosen, the state is transformed into
$$
\sum_x \frac{1}{2^{m/2}} \ket{0^m,x,\Enc_k(x)} = \ket{0^m} \otimes \sum_x \frac{1}{2^{m/2}} \ket{x,\Enc_k(x)}.
$$
\end{itemize}
Notice that in the second case we have a fully entangled state between the second and the third register. At this point, \A does the following:
\begin{enumerate}
\item measures (traces out) the third register;
\item applies again $H$ to the second register;
\item measures the second register;
\item outputs $b'=1$ iff the outcome of this last measurement is $0^m$, else outputs~$0$.
\end{enumerate}
In fact, if $b=0$, then the second register is left untouched: By applying again the Hadamard transformation it will be reset to the state $\ket{0^m}$, and a measurement on this state will yield $0^m$ with probability 1. If $b=1$ instead, tracing out one half of a fully entangled state results in a complete mixture in the second register. Applying a Hadamard transform and measuring in the computational basis necessarily gives a fully random outcome, and hence outcome $0^m$ only with probability $\frac{1}{2^m}$, which is negligible in $n$, because $m = \poly{n}$.
\qed
\end{proof}

Theorem~\ref{thm:BZatt} implies that the fqIND-qCPA notion is too strong. In order to weaken it, the following notion of indistinguishability under adaptively chosen quantum plaintext attacks was introduced:

\begin{definition}[IND-qCPA \cite{Boneh2013}]
A secret-key encryption scheme is said to be {\em IND-qCPA secure} if the success probability of any quantum \ppttxt adversary winning the game defined by qCPA learning phases and a classical IND challenge phase is at most negligibly close (in $n$) to $1/2$.
\end{definition}

In this definition, the CPA queries are allowed to be quantum, but the challenge query is required to be classical. It has been shown that, under standard computational assumptions, IND-qCPA is strictly stronger than IND-CPA: 

\begin{theorem}[IND-CPA $\not\Rightarrow$ IND-qCPA {\cite[Theorem 4.8]{Boneh2013}}]\label{thm:separationBZ}
If classically secure PRFs exist and order-finding in prime groups is classically hard, then 
there exists an encryption scheme $\E$ which is IND-CPA secure, but not IND-qCPA secure.
\end{theorem}

\section{New Notions of Quantum Indistinguishability}\label{sec:qindqcpa}

IND-qCPA might be viewed as classical indistinguishability (IND) under a quantum chosen plaintext attack (qCPA). The authors in~\cite{Boneh2013} resorted to this definition in order to overcome their impossibility result on one seemingly natural notion of quantum indistinguishability (fqIND-qCPA) which turned out to be too strong. This raises the question whether IND-qCPA is the only possible quantum indistinguishability notion (and hence no classical encryption scheme can achieve indistinguishability of ciphertext superpositions) or if there exists a stronger notion which can be achieved. 

In this section we show that by defining fqIND-qCPA, there are many choices which are made implicitly, and that on the other hand there exist other possible quantum indistinguishability notions. We discuss these choices spanning a binary `security tree' of possible notions. Afterwards, we obtain a small set of candidate notions, eliminating those that are either ill-posed or unachievable because of the BZ attack from Theorem~\ref{thm:BZatt}. In all these notions, we implicitly assume `quantum CPA learning phases', as in the case of IND-qCPA. However, we limit the discussion in this section to the design of a quantum challenge phase. In the end, we select a suitable `qIND-'notion amongst all the possible candidate ones.

\subsection{The `Security Tree'}\label{sec:tree}
To define a general notion of indistinguishability in the quantum world, we have to consider many different distinctions for possible candidate models. For example, can we rule out certain forms of entanglement? How? Does the adversary have complete control over the challenger device? Each of these distinctions leads to a fork in a `security-model binary tree'. We analyze every `leaf' of the tree\footnote{We do not rule out that some of them might eventually lead to the same model.}. Some of them lead to unreasonable or ill-posed models, some of them yield unachievable security notions, and others are analyzed in more detail.

\subheading{Game model: Oracle $(\oracle)$ vs. Challenger $(\C)$.} This distinction decides how the game, and especially the challenge phase, is implemented. In the classical world, the following two cases are equivalent but in the quantum world they differ. In the {\em oracle} model, the adversary $\A$ gets oracle access to encryption and challenge oracles, i.e., he plays the game by performing calls to unitary gates $\oracle_1,\ldots,\oracle_q$. In this case \A is modeled as a quantum circuit which implements a sequence of unitary gates $U_0,\ldots,U_q$, intertwined by calls to the $\oracle_i$'s. Given an input state $\ket{\varphi}$, the adversary therefore computes the state:
$$
U_q \oracle_q \ldots U_1 \oracle_1 U_0 \ket{\varphi}.
$$

The {\em structure} of the {\em oracle} gates $\oracle_i$ itself is unknown to $\A$, who is only allowed to apply them in a black-box way. The fqIND notion uses this model.

\label{rewind}In what we call the {\em challenger} model instead, the game is played against an {\em external (quantum) challenger}. Here, \A is a quantum circuit which shares a quantum register (the communication channel) with another quantum circuit \C. The main difference is that in this case we can also consider what happens if \C has additional input or output lines out of \A's control. Moreover, \A does not automatically gain access to the inverse (adjoint) of quantum operations performed by \C, and \C cannot be `rewound' by the adversary, which would be far too powerful possibilities. This scenario also covers the case of `unidirectional' state transmission, i.e., when qubits are sent over a quantum channel to another party, and they are not available afterwards until that party sends them back. Regardless, in security proofs in the $(\C)$ model, it is still allowed for an external entity (e.g. a simulator, or a reduction) to rewind the joint circuit composed by adversary and challenger together, if need be. However, we are not aware of any known reduction involving rewinding in this form for encryption schemes in the quantum world.

In order to keep consistency with this choice of the model, when also considering qCPA queries, we implicitly assume the same access mode to the $\Enc_k$ oracle as in the qIND game. That is, if we are in the $(\oracle)$ scenario, during the qCPA phase \A has quantum oracle access to $\Enc_k$. In the $(\C)$ case, instead, superposition access to $\Enc_k$ is provided to \A by an external challenger.

At first glance, the $(\oracle)$ model intuitively represents the scenario where \A has almost complete control of some encryption device, whereas the $(\C)$ model is more suited to a `network' scenario where \A wants to compromise the security of some external target.

\subheading{Plaintexts: quantum states $(Q)$ vs. classical description $(c)$.} In the $(Q)$ model, the two $m$-qubit plaintexts chosen by \A for the challenge template can be arbitrary (BQP-producible) quantum states and can be entangled with each other and other states. In the $(c)$ model, instead, \A is only allowed to choose {\em classical descriptions} of two $m$-qubit quantum states according to Definition~\ref{def:classrepr}, thus being only allowed to send classical information to \C: the challenger \C will read the states' descriptions and will build one of the two states depending on his challenge bit $b$.

In classical models, there is no difference between sending a description of a message or the message itself. In the quantum world, there is a big difference between these two cases, as the latter allows \A to establish entanglement of the message(s) with other registers. This is not possible when using classical descriptions. It might intuitively appear that the $(Q)$ model (considered for the fqIND-qCPA notion) is more natural. However, the $(c)$ scenario models the case where \A is well aware of the message that is encrypted, but the message is not constructed by \A himself. Giving \A the ability to choose the challenge messages for the IND game models the worst case that might happen: \A knows that the ciphertext he receives is the encryption of one out of the two messages that he can distinguish best. This closely reflects the intuition behind the classical IND notions: in that game, the adversary is allowed to send the two messages not because in the real world he would be allowed to do so, but because we want to achieve security even for the best possible choice of messages from the adversary's perspective. Hence, the $(c)$ model is a valid alternative. Will further discuss the difference between these two models later.

\subheading{Relaying of plaintext states: Yes $(Y)$ vs. No $(n)$.} If \C is {\em not relaying} $(n)$, this means that the two plaintext states chosen by \A will not be `sent back' to \A (in other words: their registers will not be available anymore to \A after the challenge encryption). In circuit terms, this means that at the beginning of the game, \C will have (one or two) ancilla registers in his internal (private) memory. During the encryption phase, \C will swap these register(s) with the content of the original plaintext register(s), hence transferring their original content outside of \A's control.

If the challenger is relaying $(Y)$ instead, this means that the two plaintext states will be left in the original register (or channel), and may be accessed by \A at any moment. This is the model considered for fqIND.

Again, the $(Y)$ case is more fitting to those cases where \A `implements locally' the encryption device and has almost full control of it, whereas the $(n)$ case is more appropriate when the game is played against some external entity which is not under \A's control. This is a rather natural assumption, for example, when states are sent over some quantum channel and not returned. We stress that this distinction in relaying is not trivial: it is not possible for \A, in general, to simulate relaying by keeping internal states entangled with the plaintexts. As an example, consider the attack in Theorem~\ref{thm:BZatt}: it is easy to see that this cannot be performed without relaying.

\subheading{Type of unitary transformation: $(1)$ vs. $(2)$.} In quantum computing, the `canonical' way of evaluating a function $f(x)$ in superposition is by using an auxiliary register:
$$
\sum_{x,y} \alpha_{x,y} \ket{x,y} \mapsto \sum_{x,y}  \alpha_{x,y} \ket{x,y \xor f(x)}.
$$
This way ensures that the resulting operator is invertible, even if $f$ is not. We call this {\em type-$(1)$ transformations}: if $\Enc_k$ is an encryption mapping $m$-bit plaintexts to $\ell$-bit ciphertexts, the resulting operator in this case will act on $m+\ell$ qubits in the following way:
$$
\sum_{x,y} \alpha_{x,y} \ket{x,y} \mapsto \sum_{x,y} \alpha_{x,y} \ket{x,y \xor \Enc_k(x)},
$$
where the $y$'s are ancillary values. This approach is also used for fqIND.

In our case, though, we do not consider arbitrary functions, but encryptions, which act as {\em bijections} on some bit-string spaces (assuming that the randomness is treated as an input.) Therefore, provided that the encryption does not change the size of a message, the following transformation is also invertible:
\begin{equation}\label{eq:UEnc2}
\sum_{x} \alpha_{x} \ket{x} \mapsto \sum_{x} \alpha_{x} \ket{\Enc_k(x)}.
\end{equation}
For the more general case of arbitrary message expansion factors, we will consider transformations of the form:
$$
\sum_{x,y} \alpha_{x,y} \ket{x,y} \mapsto \sum_{x,y} \alpha_{x,y} \ket{\varphi_{x,y}},
$$
where the length of the ancilla register is $|y|\!=\!|\Enc_k(x)|-|x|$ and ${\varphi_{x,0}\!= \Enc_k(x)}$ for every $x$ -- i.e., initializing the ancilla $y$ register in the $\ket{0}$ state produces a correct encryption, which is what we expect from an honest quantum executor. One might ask what happens if the ancilla is not initialized to $0$, and we leave the general case of arbitrary ancillas manipulation as an interesting open problem, but we stress the fact that this behavior is not considered in the case of honest parties. We call these {\em type-$(2)$ transformations}\footnote{These are called {\em minimal quantum oracles} in~\cite{KKVB02}.}.

Notice that, in general, type-$(1)$ and type-$(2)$ transformations are very different: having quantum oracle access to a type-$(2)$ unitary $U^{(2)}_\Enc$ and its adjoint also gives access to the related type-$(2)$ {\em decryption oracle} $U^{(2)}_\Dec : \sum_x \alpha_x \ket{\Enc_k(x)} \mapsto \sum_x \alpha_x \ket{x}$. In fact, notice that $(U^{(2)}_\Enc)^\dagger = U^{(2)}_\Dec$, while the adjoint of a type-$(1)$ encryption operator, $(U^{(1)}_\Enc)^\dagger$, is generally {\em not} a type-$(1)$ decryption operator. In particular, type-$(2)$ operators are `more powerful' in the sense that knowledge of the secret key is required in order to build any efficient quantum circuit implementing them. However, we stress the fact that whenever access to a decryption oracle is allowed, the two models are completely equivalent, because then we can simulate a type-(2) operator by using ancilla qubits and `uncomputing' the resulting garbage lines (see Figure~\ref{fig:equiv12}) (as we will see, this will be the case for the challenger in our qIND notion).
\begin{figure}[h]
\center
 \includegraphics[width=\textwidth,keepaspectratio]{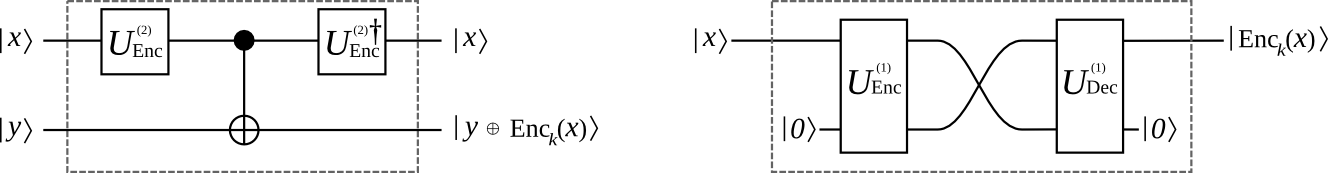}
 \caption{Equivalence between type-(1) and type-(2) in the case of \mbox{$1$-qubit} messages. Left: building a \mbox{type-(1)} encryption oracle by using a \mbox{type-(2)} encryption oracle (and its inverse) as a \mbox{black-box}. Right: building a \mbox{type-(2)} encryption oracle by using \mbox{type-(1)} encryption and decryption oracles as \mbox{black-boxes}.\label{fig:equiv12}}
\end{figure}

\subsection{Analysis of the models}\label{analysis}
By considering these 4 distinctions in the security tree we have $2^4=16$ possible candidate models to analyze. We label each of these candidate models by appending each one of the $4$ labels of every tree branch in brackets. Clearly, 16 different definitions of quantum indistinguishability is too much, but luckily most of these are unreasonable or unachievable. To start with, we can ignore the following:

\subheading{Leaves of the form $(\oracle c \ldots)$.} In the \oracle scenario, the oracle is actually a quantum gate inside \A's quantum circuitry. Therefore \A has the capability of querying the oracle on states which are possibly entangled with other registers kept by \A itself.

\subheading{Leaves of the form $(\oracle Q n \ldots)$.} Again, the oracle is a gate which has no internal memory to store and keep the plaintext states sent by \A.

\subheading{Leaves of the form $(\ldots Y 2 )$.} Relaying is not taken into account in type-$(2)$ transformations. In these transformations, to some extent, one of the two plaintext registers is {\em always} relayed (after having been `transformed' into a ciphertext). If the other plaintext was to be relayed as well, this would immediately compromise indistinguishability (because one of the two states would be modified and the other not, and both of them would be handed over to \A).

\label{finalists}Excluding these options leaves us with 7 models, but it is easy to see that 3 of them are unachievable because of the attack from Theorem~\ref{thm:BZatt}. This is the case for $(\oracle Q Y 1)$ (which is exactly fqIND-qCPA), $(\C Q Y 1)$, and $(\C c Y 1)$. Of the remaining 4, notice that $(\C Q n 1)$ and $(\C c n 1)$ are equivalent to the IND-qCPA notion from \cite{Boneh2013}. The reason is that from \A's perspective, a non-relaying \C is indistinguishable from a \C tracing out (measuring) the plaintext register (otherwise \A and \C could communicate faster than light). This measuring operation would make the ciphertext collapse into a single (classical) ciphertext. And since tracing out the challenge register and applying the type-(1) operator $U^{(1)}_{\Enc}$ commute, one can consider (without loss of generality) the case that \A himself first measures the plaintext register, and then initiates a classical IND query with \C , therefore recovering a classical definition of IND challenge query\footnote{However, we stress that this interpretation is not entirely correct. In fact, one might consider composition scenarios where the IND query is just an intermediate step, and the plaintext and ciphertext registers are reunited at some later step. In such scenarios, not relaying would not be equivalent to measuring. We ignore such considerations in this work, and leave the general case of composable security as an interesting open question.}. Therefore, using any of $(\C Q n 1)$ or $(\C c n 1)$ would lead to a weaker notion of quantum indistinguishability. Since we are interested in achieving stronger notions, we will hence consider the more challenging  scenarios $(\C Q n 2)$ and $(\C c n 2)$.

This argument also leads to the following interesting observation. Ultimately, whether a challenger (or encryption device) performs type-(1) or type-(2) operations depends on its {\em architecture} which we cannot say anything about - we will focus on the $(\ldots 2)$ models in order to be on the `safe side', as they lead to security notions which are harder to achieve. In order to design a secure encryption device, it is good advice to avoid the possibility that it can be accessed in type-(2) mode. For such a device, it would be sufficient to provide IND-qCPA security, which is weaker and therefore easier to achieve. Clearly, providing guidelines on how to construct encryption devices resilient to type-(2) access lies outside the scope of this work.

\subsection{qIND}\label{sec:qIND}
At this point we are left with only two candidate notions: $(\C c n 2)$ and $(\C Q n 2)$. 
From now on we will denote them as {\em `quantum indistinguishability of ciphertexts' (qIND)} and {\em `general quantum indistinguishability of ciphertexts' (gqIND)} resp., and we summarize the resulting challenge phases as follows.

\subheading{Quantum IND challenge phase (qIND):} \A chooses two quantum states $\rho_0,\rho_1$ having efficient (poly-sized) classical descriptions, and sends to \C a challenge template consisting of these two classical descriptions according to Definition~\ref{def:classrepr}. \C samples a bit $b$ and replies to \A with the state obtained by applying the type-(2) operator $U^{(2)}_{\Enc_k}$ as defined in~\eqref{eq:UEnc2} to $\rho_b$. \A's goal is to output $b$.

\subheading{General Quantum IND challenge phase (gqIND):} \A chooses two quantum states $\rho_0,\rho_1$, and sends them to \C . \C samples a bit $b$, discards (traces out) $\rho_{1-b}$, and replies to \A with the state obtained by applying the type-(2) operator $U^{(2)}_{\Enc_k}$ as defined in~\eqref{eq:UEnc2} to $\rho_b$. \A's goal is to output $b$.

Using these challenge phases and the notion of a qCPA learning phase, we define qIND-qCPA and gqIND-qCPA as follows.

\begin{definition}[qIND-qCPA]\label{def:qIND}
A secret-key encryption scheme is said to be qIND-qCPA secure 
if the success probability of any quantum probabilistic polynomial time adversary winning the game defined by qCPA learning phases and the qIND challenge phase above is at most negligibly close (in $n$) to $1/2$.
\end{definition}

\begin{definition}[gqIND-qCPA]\label{def:gqIND}
A secret-key encryption scheme is said to be gqIND-qCPA secure
if the success probability of any quantum probabilistic polynomial time adversary winning the game defined by qCPA learning phases and the gqIND challenge phase above is at most negligibly close (in $n$) to $1/2$.
\end{definition}

Since we mainly consider type-(2) transformations from now on, we will overload notation and also use $U_{\Enc_k}$ to denote the type-(2) encryption operator.

\begin{theorem}[gqIND-qCPA $\Rightarrow$ qIND-qCPA]\label{theo:qind-gqind}
Let $\E$ be a gqIND-qCPA secure symmetric-key encryption scheme. Then $\E$ is also qIND-qCPA secure.
\end{theorem}

\label{cmodel}The reason is that quantum states admitting an efficient classical description (used in qIND) are just a special case of arbitrary quantum plaintext states (used in gqIND). Despite this implication, we will mainly focus on the qIND notion in the following, and we will use the gqIND notion only as a comparison to other existing notions. 
The main reason for this choice is that in the context of classical encryption schemes resistant to superposition quantum access, we believe that it is important to not lose focus of what the capabilities of a `reasonable' adversary should be. Namely, recall the following classical IND argument: {\em allowing the adversary to send plaintexts to the challenger is equivalent to the fact that indistinguishability must hold even for the most favorable case from the adversary's perspective}. Such an argument does {\em not} hold anymore quantumly. In fact, the $(Q)$ model considered in gqIND presents the following issues:
\begin{itemize}
\item it allows entanglement between the adversary and the challenger: \A could prepare a state of the form $\rho_{AB} = \frac{1}{\sqrt{2}}\ket{00}+\frac{1}{\sqrt{2}}\ket{11}$, sending $\rho_A$ as a plaintext but keeping $\rho_B$;
\item it allows the adversary to create certain non-reproduceable states. For example, consider the state $\ket{\psi} = \sum_{x \in X} \frac{1}{\sqrt{|X|}} \ket{x,h(x)}$, where $h$ is a collision-resistant hash function. \A could measure the second register, obtaining a random outcome $y$, and knowing therefore that the remaining state is the superposition of the preimages of $y$, $\ket{\psi_y} = \sum_{x \in X :h(x)=y} \frac{1}{\sqrt{|\set{x \in X:h(x)=y}|}} \ket{x}$. \A could then use $\ket{\psi_y}$ as a plaintext in the challenge phase, but note that $\A$ cannot reproduce $\ket{\psi_y}$ for a given value $y$.
\end{itemize}
Both of the above examples are not reasonable in our scenario. Entanglement between \A and \C represents a sort of `quantum watermarking' of messages, which goes beyond what a meaningful notion of indistinguishability should achieve. Knowledge of intermediate, unpredictable measurements also renders \A too powerful, because it gives \A access to information not available to \C itself - e.g., in the example above \C would not even know the value of $y$. As it is \C who prepares the state to be encrypted, it is reasonable to assume that it is \C who should know these intermediate measurements, not \A. In the example above, what \A could see instead (provided he knows the circuit generating the state, as we assume in qIND) is that the plaintext is a mixture $\Psi=\sum_y \psi_y$ for all possible values of $y$.

The possibility offered by gqIND of allowing the adversary to play the IND game with arbitrary states is certainly elegant from a theoretical point of view, but from the perspective of the quantum security of the kind of schemes we are considering, it is too broad in scope. The $(c)$ model used in qIND, on the other hand, inherently provides guidelines and reasonable limitations on what a quantum adversary can or cannot do. Also, qIND is often easier to deal with: notice that in the $(c)$ model, unlike in the $(Q)$ model, \A always receives back an unentangled state from a challenge query. In security reductions, this means that we can more easily simulate the challenger, and that we do not have to take care of measures of entanglement when analyzing the properties of quantum states - for example, indistinguishability of states can be shown by only resorting to the {\em trace norm} instead of the more general {\em diamond norm}.

Furthermore, it is important to notice that all our new results in Section~\ref{sec:constrimp} are unaffected by the choice of either qIND or gqIND. Our impossibility result from Theorem~\ref{thm:impossibility} holds for qIND, and hence also for gqIND because of Theorem~\ref{theo:qind-gqind}. On the other hand, the security proof of Construction~\ref{const:extension} (Theorem~\ref{thm:extension}) is given for gqIND, and holds therefore also for qIND. In fact, it remains unclear whether a separation between qIND and gqIND can be found at all in the realm of classical encryption schemes. We leave this as an interesting open question.

\label{BJcomparison}Finally, we note that the q-IND-CPA-2 indistinguishability notion for secret-key encryption of quantum messages introduced by Broadbent and Jeffery~\cite[Appendix~B]{BJ15} resembles our gqIND notion, and it is in fact equivalent to it in the case that the encryption operation is a symmetric-key classical functionality operating in type-(2) mode.

\begin{theorem}[gqIND-qCPA $\Leftrightarrow$ q-IND-CPA-2]\label{theo:gqind-qindcpa2}
Let $\E$ be a symmetric-key encryption scheme. Then $\E$ is gqIND-qCPA secure iff $\E$ is q-IND-CPA-2 secure.
\end{theorem}

A proof of the above theorem can be found in Appendix~\ref{app:BJ15equiv}. A generalization of q-IND-CPA-2 to arbitrary quantum encryption schemes, together with equivalent notions of quantum semantic security, was given and analized in~\cite{Alagic+16arxiv}. All these security notions are given in the context of `fully quantum encryption', in the sense that the encryption schemes considered in~\cite{BJ15} and~\cite{Alagic+16arxiv} are arbitrary quantum circuits acting natively on quantum data, while in this work we consider the quantum security of classical encryption schemes. The fully quantum homomorphic schemes which are shown to be secure in~\cite{BJ15}, and the other quantum encryption schemes shown to be secure in~\cite{Alagic+16arxiv}, do not fall into the category of \emph{classical} encryption schemes which we are studying here. On the other hand, as Theorem~\ref{thm:extension} shows, our Construction~\ref{const:extension} is the first known example of a classical symmetric-key encryption scheme which is secure even against these kinds of `fully quantum' security notions.

\section{New Notions of Quantum Semantic Security}

In this section, we initiate the study of suitable definitions of semantic security in the quantum world. As in the classical case, we are particularly interested in notions that can be proven equivalent to some version of quantum indistinguishability. So these definitions actually describe the {\em semantics} of the equivalent IND notions. As in the classical case, we present these notions in the non-uniform model of computation.

Working towards a quantum SEM notion, we restrict our analysis to the SEM challenge phase. For the learning phase, we stick to the `qCPA learning phase', as in Definition~\ref{def:quantumEncryptionOracle}, where the adversary has access to a quantum encryption oracle. In the end, we give a definition for quantum semantic security under quantum chosen-plaintext attacks (qSEM-qCPA) which we later prove equivalent to qIND-qCPA, thereby adding semantics to our qIND-qCPA notion.

\subsection{Classical Semantic Security under Quantum CPA}\label{qasemquestion}

As a first notion of semantic security in the quantum world, we consider what happens if, like in the IND-qCPA notion, we stick to the classical definition but we allow for a quantum chosen-plaintext-attack phase. The definition uses a SEM-qCPA game that is obtained by combining qCPA learning phases with a classical SEM challenge phase as defined in Section~\ref{sec:preliminaries}. As in the classical case, \A's success probability is compared to that of a simulator \sim that plays in a reduced game: \sim gets no learning phase and during the challenge phase it only receives the advice $h_m(x)$, not the ciphertext. 

\begin{definition}[SEM-qCPA]\label{def:semqcpa} A secret-key encryption scheme is called SEM-qCPA-secure if for every quantum polynomial-time machine \A, there exists a quantum polynomial-time machine \sim such that the challenge templates produced by \sim and \A are identically distributed and the success probability of \A winning the game defined by qCPA learning phases and a SEM challenge phase is negligibly close (in $n$) to the success probability of \sim winning the reduced game. 
\end{definition}

\subheading{Spoiler.} 
It is easy to see that the SEM-qCPA notion of semantic security is equivalent to IND-qCPA, see Theorem~\ref{theo:qcpa}.

\vspace{0.2in}

In Appendix~\ref{app:qaSEM} we discuss what happens if one also allows quantum advice states in this scenario, and why this option would not add anything meaningful.

\subsection{Quantum Semantic Security}\label{sec:qsem}

Here we define {\em quantum semantic security under chosen-plaintext attacks}~(qSEM-qCPA). As in the classical case, we want the definition of semantic security to formally capture what we intuitively understand as a strong security notion. In the quantum case, there are several choices to be made. We start by giving our formal definition of quantum semantic security, and justify our choices afterwards.

\subheading{Quantum SEM (qSEM) challenge phase:} \A sends to \C a challenge template consisting of classical decriptions of 
\begin{itemize}
\item a quantum circuit $G_m$ taking $\poly{n}$-bit classical input and outputting $m$-qubit plaintext states,
\item a quantum circuit $h_m$ taking $m$-qubit plaintexts as input and outputting $\poly{n}$-qubit advice states,
\item a quantum circuit $f_m$ taking $m$-qubit plaintexts as input and outputting $\poly{n}$-qubit target states.
\end{itemize}

The challenger \C samples $y \rand \bin^\poly{n}$ and computes two copies of the plaintext $\rho_y = G_m(y)$. One is used to compute auxiliary information $h_m(\rho_y)$ and one to compute the ciphertext $U_{\Enc_k} \, \rho_y \, U_{\Enc_k}^\dag$. \C then replies with the pair $\left( U_{\Enc_k} \, \rho_y \, U_{\Enc_k}^\dag, h_m(\rho_y) \right)$. \A's goal is to output $f_m(\rho_y)$. We say that $\A$ {\em wins the qSEM-qCPA game} if no quantum polynomial-time distinguisher can distinguish \A's output from the target state $f_m(\rho_y)$ with non-negligible advantage.

In the reduced game, \sim receives no encryption, but only the auxiliary information $h_m(\rho_y)$ from \C. Analogously to the above case, \sim {\em wins the qSEM-qCPA game} if no quantum polynomial-time distinguisher can distinguish \sim's output from the target state $f_m(\rho_y)$ with non-negligible advantage.

\begin{definition}[qSEM-qCPA]\label{def:qsemqcpa} A secret-key encryption scheme is called qSEM-qCPA-secure if for every quantum polynomial-time machine \A, there exists a quantum polynomial-time machine \sim such that the challenge templates produced by \sim and \A are identically distributed and the success probability of \A winning the game defined by qCPA learning phases and a qSEM challenge phase is negligibly close (in $n$) to the success probability of \sim winning the reduced game. 
\end{definition}

When defining quantum semantic security, we have to deal with several issues: First, we have to define how the plaintext  distribution is described. In the classical definition, the distribution is produced by a (classical) circuit $G_m$ running on uniform input bits. We take the same approach here, but let $G_m$ output $m$-qubit plaintexts.

The second question is how to define the advice function. While the input should be the plaintext quantum state $\rho_y$, the output could be either quantum or classical. We decided to allow quantum advice as it leads to a more general model and it includes classical outputs as a special case. In order for the challenger to compute both the encryption of the plaintext state $\rho_y$ and the advice state $h_m(\rho_y)$ without violation of the no-cloning theorem, we exploit how we generate the message state. We simply run $S_m$ twice on the same classical randomness $y$ to generate two copies of the plaintext state $\rho_y$. Another option would have been to allow for entanglement between the plaintext message $\rho_y$ and the advice state $h_m(\rho_y)$. Allowing such entanglement would model side-channel information the attacker could obtain, for instance by learning the content of some internal register of the attacked device. However, the resulting notion would not be equivalent with qIND-qCPA anymore, because in qIND-qCPA, the challenge plaintexts are provided by their classical descriptions and can therefore not be entangled with the attacker.

Third, we have chosen to model the target function $f_m$ in the same way as the advice function $h_m$, i.e.\ we allow arbitrary quantum circuits that might output quantum states. The reasoning behind allowing quantum output is again to use the strongest possible, most general model. Allowing quantum output however leads to the problem that, in general, we cannot physically test anymore if an adversary \A outputs exactly the result of the target function $f_m(\rho_y)$. One option would be to require \A's output to be close to $f_m(\rho_y)$ in terms of their trace distance. But two quantum states can be quantum-polynomial-time indistinguishable even if their trace distance is large\footnote{Think of two different classical ciphertexts which are encrypted using a quantum-computationally secure encryption scheme. Then, the ciphertext states are orthogonal (and hence their trace distance is maximal), but they are computationally indistinguishable.}. Since we are only interested in computational security notions, we solve this problem by requiring QPT indistinguishability as success condition for winning the SEM game.

\subheading{Spoiler.} 
Our qSEM-qCPA notion of semantic security is equivalent to qIND-qCPA, and unachievable for those schemes which leave the size of the message unchanged (like most block ciphers), see Section~\ref{imposs}.

\section{Relations}\label{relations}

In this section we show relations between our new notions of indistinguishability and semantic security in the quantum world. It is already known~\cite{GM84,Goldreich2004} that classically, IND-CPA and semantic security are equivalent. Our goal is to show a similar equivalence for our new notions, plus to show a hierarchy of equivalent security notions. Our results are summarized in Figure \ref{fig:relations}.

\begin{figure}[h]
 \includegraphics[width=\textwidth,keepaspectratio]{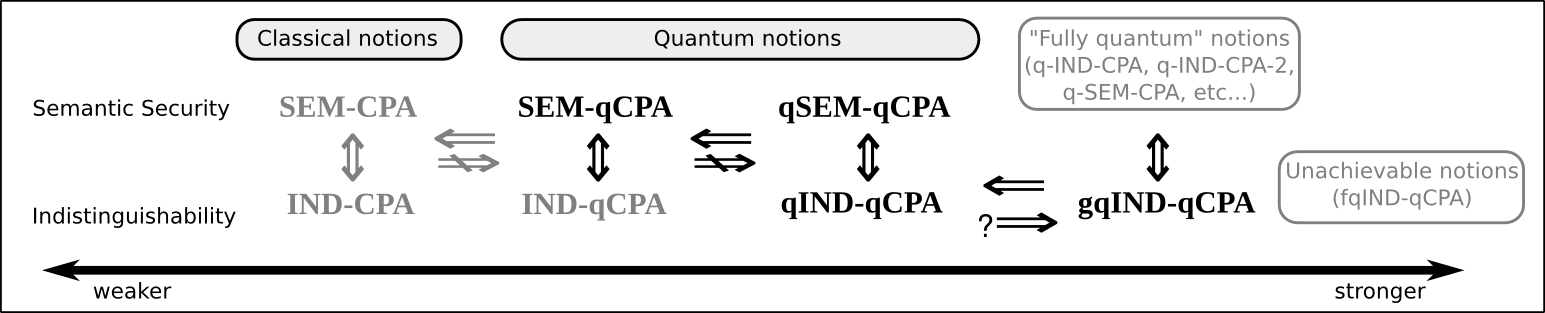}
 \caption{The relations between notions of indistinguishability and semantic security in the quantum world (previously known results in gray.)\label{fig:relations}}
\end{figure}

We start by proving equivalence between IND-qCPA and SEM-qCPA.

\begin{theorem}[IND-qCPA $\Leftrightarrow$ SEM-qCPA]\label{theo:qcpa}
Let $\E$ be a symmetric-key encryption scheme. Then $\E$ is IND-qCPA secure iff $\E$ is SEM-qCPA secure.
\end{theorem}

We split the proof of Theorem~\ref{theo:qcpa} into two propositions -- one per direction. They closely follow the proofs for the classical case (see~\cite[Proof of Th.~5.4.11]{Goldreich2004}), we recall them as they work as guidelines for the following proofs.

\begin{proposition}[IND-qCPA $\Rightarrow$ SEM-qCPA]\label{prop:indq:semq}
\end{proposition}

\begin{proposition}[SEM-qCPA $\Rightarrow$ IND-qCPA]\label{prop:semq:indq}
\end{proposition}

\begin{proof}[of Proposition~\ref{prop:indq:semq} -- Sketch.]
The idea of the proof is to hand \A's circuit as non-uniform advice to the simulator \sim. \sim runs \A's circuit and impersonates the challenger \C by generating a new key and answering all of \A's queries using this key. When it comes to the challenge query, \sim encrypts the $1\ldots1$ string of the same length as the original message. It follows from the indistinguishability of encryptions that the adversary's success probability in this game must be negligibly close to its success probability in the real semantic-security game, which concludes the proof. The only difference in the -qCPA case is that \A and \sim are quantum circuits, and that \sim has to emulate the quantum encryption oracle instead of a classical one.\qed
\end{proof}

\begin{proof}[of Proposition~\ref{prop:semq:indq}]
We recall here the full proof as it is short. Assume there exists an efficient distinguisher \A against the IND-qCPA security of \E. Then we show how to construct an oracle machine $\MA$ that has access to \A and breaks the SEM-qCPA security of the scheme. $\MA$ runs \A , emulating the quantum encryption oracle by simply forwarding all the qCPA queries to its own oracle. As \A executes an IND challenge query on $m$-bit messages $(x_0,x_1), \MA$ produces the SEM template $(G_m,h_m,f_m)$ with $G_m$ describing the uniform distribution over $\set{x_0,x_1}, h_m=1^n$ (or any other function such that $h_m(x_0)=h_m(x_1)$), and $f_m$ a function that fulfills $f_m(x_0)=0$ and $f_m(x_1)=1$ (i.e., the distinguishing function). Then $\MA$ performs a SEM challenge query with this template, and given challenge ciphertext $c$, uses it to answer \A's query. If, at that point, \A performs more qCPA queries, $\MA$ answers again by forwarding all these queries to its own oracle. Finally, $\MA$ outputs \A's output. As \A distinguishes encryptions of $x_0$ and $x_1$ with non-negligible success probability, \A will return the correct value of $f_m$ with recognizably higher probability than guessing. As $h_m$ is independent of the encrypted message, no simulator can do better than guessing. Hence, \MA has a non-negligible advantage to output the right value of $f_m$. \qed
\end{proof}

Next, we show equivalence between qIND-qCPA and qSEM-qCPA.

\begin{theorem}\label{theo:qequi}[qIND-qCPA $\Leftrightarrow$ qSEM-qCPA]
Let $\E$ be a symmetric-key encryption scheme. Then $\E$ is qIND-qCPA secure iff $\E$ is qSEM-qCPA secure.
\end{theorem}

Again, we split the proof of Theorem~\ref{theo:qequi} into two propositions.

\begin{proposition}\label{prop:qindq:qsemq}[qIND-qCPA $\Rightarrow$ qSEM-qCPA]
\end{proposition}

\begin{proposition}\label{prop:qsemq:qindq}[qSEM-qCPA $\Rightarrow$ qIND-qCPA]
\end{proposition}

\begin{proof}[of Proposition~\ref{prop:qindq:qsemq} -- Sketch.]
The proof follows that of Proposition~\ref{prop:indq:semq}, with some careful observations. Since \A is a QPT adversary against the qSEM-qCPA game, \A's circuit has a short classical representation $\xi$. So \sim gets $\xi$ as non-uniform advice and hence can implement and run \A. The simulator \sim simulates \C for \A by generating a new key and answering all of \A's qCPA queries. When it comes to the challenge query, \A produces a qSEM template, which \sim forwards to the real \C. Then \sim forwards \C's reply, plus a bogus encrypted state (e.g., $U_{\Enc_k} \, \ket{1\ldots1}$), to \A. If at this point \A outputs a state $\varphi$ which can be efficiently distinguished from the correct $f_m(\rho_y)$ computed by the real \C, we would have an efficient distinguisher against the qIND-qCPA security of the scheme. Hence, \A's (and therefore also \sim's) output must be indistinguishable from $f_m(\rho_y)$ for any QPT distinguisher, which concludes the proof.\qed
\end{proof}

\begin{proof}[of Proposition~\ref{prop:qsemq:qindq}]
This is also similar to the proof of Proposition~\ref{prop:semq:indq}. Given an efficient distinguisher \A for the qIND-qCPA game, our adversary for the qSEM-qCPA game is an oracle machine $\MA$ running \A and acting as follows. Concerning \A's qCPA queries, as usual $\MA$ just forwards everything to the qSEM-qCPA challenger \C. When \A performs a challenge qIND query by sending the classical descriptions of two states $\varphi_0$ and $\varphi_1$, $\MA$ prepares the qSEM template $(G_m,h_m,f_m)$, with $G_m$ outputing $\varphi_0$ for half of the possible $y$ values and $\varphi_1$ for the other half, $h_m(\rho_y)=1^n$, and $f_m$ the identity map $f_m(\rho_y)=\rho_y$. Then $\MA$ performs a qSEM challenge query with this template. Given challenge ciphertext state $U_{\Enc_k} \, \varphi_b \, U_{\Enc_k}^\dag$ (for $b\in\bin$), he forwards it as an  answer to \A's challenge query. As \A distinguishes $U_{\Enc_k} \, \varphi_0 \, U_{\Enc_k}^\dag$ from $U_{\Enc_k} \, \varphi_1\, U_{\Enc_k}^\dag$ with non-negligible success probability, \A returns the correct value of $b$ with non-negligible advantage over guessing. Then $\MA$, having recorded a copy of the classical descriptions of $\varphi_0$ and $\varphi_1$, is able to compute the state $f_m(\varphi_b)$ exactly, and consequently win the qSEM-qCPA game with non-negligible advantage. As $h_m$ generates the same advice state $h_m(\rho_y)=1^n$ independently of the encrypted message, no simulator can do better than guessing the plaintext. This concludes the proof.\qed
\end{proof}

\vspace{0.2in}

Finally, we show the separation result between the two classes of security we have identified (we show it between IND-qCPA and qIND-qCPA). This shows that qIND-qCPA (and equivalently qSEM-qCPA) is a strictly stronger notion than IND-qCPA (which is equivalent to SEM-qCPA).

\begin{theorem}[IND-qCPA $\nRightarrow$ qIND-qCPA]\label{theo:separ}
There exists a symmetric-key encryption scheme $\E$ which is IND-qCPA secure but not qIND-qCPA secure.
\end{theorem}

\begin{proof}[of Theorem~\ref{theo:separ}]
The scheme we use as a counterexample is the one from \cite{Goldreich2004}(Construction 5.3.9). It has been proven in \cite{Boneh2013} that this scheme is IND-qCPA secure if the used PRF is post-quantum secure. We exhibit a distinguisher \A which breaks the qIND-qCPA security of this scheme with high probability. For ease of notation we restrict to the case of single-bit messages $0$ and $1$. \A will simply choose as challenge states: $\ket{\varphi_0} = H\ket{0} = \frac{1}{\sqrt{2}}\ket{0} + \frac{1}{\sqrt{2}}\ket{1}$, and $\ket{\varphi_1} = H\ket{1} = \frac{1}{\sqrt{2}}\ket{0} - \frac{1}{\sqrt{2}}\ket{1}$. When the challenger \C applies the type-2 transformation to either of these two states, it is easy to see that in any case the state is left unchanged. This is because $U_{\Enc_k}$ just applies a permutation in the space of the basis elements, but $\ket{\varphi_0}$ and $\ket{\varphi_1}$ have the same amplitudes on all their components, except for the sign. As these two states are orthogonal, they can be reliably distinguished by the adversary \A who can then win the qIND-qCPA game with probability 1.\qed
\end{proof}

The above proof can be generalized to message states of arbitrary length, as our impossibility result in Section~\ref{imposs} shows.

\section{Impossibility and Achievability Results}\label{sec:constrimp}

In this section we show that qIND-qCPA (and equivalently qSEM-qCPA) is impossible to achieve for encryption schemes which do not expand the message (such as stream ciphers and many block ciphers, without considering the randomness part in the ciphertext). Therefore, for a scheme to be secure according to this new definition, it is necessary (but not sufficient) to increase the message size during the encryption. Interestingly, such an increase happens in most public-key post-quantum encryption schemes, like for example LWE based schemes \cite{Lindner2011} or the McEliece scheme \cite{McEliece1978}.

Then we propose a construction of a qIND-qCPA--secure symmetric-key encryption scheme. Our construction works for any (quantum-secure) pseudorandom permutation (PRP). Given that block ciphers are usually modelled as PRPs, it seems reasonable to assume that we can obtain a secure scheme when using block ciphers with sufficiently large key and block size. Hence, our construction can be used to patch existing schemes, or as a guideline in the design of quantum-secure encryption schemes from block ciphers.

\subsection{Impossibility Result}\label{imposs}

First we formally define what it means for a cipher to expand or keep constant the message size by defining the {\em core function} of a (secret-key) encryption scheme. Intuitively, the definition splits the ciphertext into the randomness and a part carrying the message-dependent information. This definition covers most encryption schemes in the literature.

\begin{definition}[Core function]\label{def:core}
 Let $(\Gen,\Enc,\Dec)$ be a secret-key encryption scheme. We call the function $f:\kSpace \times \bin^\tau \times \msgSpace \rightarrow \range$ the {\em core function} of the encryption scheme if, for some $\tau \in \NN$:
 \begin{itemize}
  \item for all $k\in\kSpace$ and $x \in \msgSpace$, $\Enc_k(x)$ can be written as $(r,f(k,r,x))$, where $r\in\bin^\tau$ is independent of the message; and
  \item there exists a function $f'$ such that for all $k\in\kSpace, r\in\bin^\tau, x\in\msgSpace$, we have: $f'(k,r,f(k,x,r)) = x$.
 \end{itemize}
\end{definition}

For example, in case of Construction~5.3.9 from~\cite{Goldreich2004} (where $Enc_k(x)$ is defined as $(r,\F_k(r)\xor x)$ for a PRF $\F$) the core function is $f(k,r,x) = \F_{k}(r) \xor x$, with $f'(k,r,z) = z \oplus F_k(r)$.

\begin{definition}[Quasi--length-preserving encryption]
We call a secret-key encryption scheme with core function $f$ {\em quasi--length-preserving} if 
$$
\forall x \in \msgSpace, r \in \bin^\tau, k \in \kSpace \Rightarrow |f(k,x,r)| = |x|,
$$
i.e., if the output of the core function has the same bit length as the message.
\end{definition}

Continuing the above example, Construction 5.3.9 from \cite{Goldreich2004} is quasi--length-preserving.

The crucial observation is the following: For a quasi--length-preserving encryption scheme, the space of possible input and (core function) output bitstrings (with respect to plaintext and ciphertext) coincide, therefore these ciphers act as permutations on this space. This means that if we start with an input state which is a superposition of {\em all} the possible basis states, all of them with the {\em same} amplitude, this state will be unchanged by the unitary type-(2) encryption operation (because it will just `shuffle' in the basis-state space amplitudes which are exactly the same).

\begin{theorem}[Impossibility Result]\label{thm:impossibility}
No quasi--length-preserving secret-key encryption scheme can be qIND secure.
\end{theorem}

\begin{proof}
Let $(\Gen,\Enc,\Dec)$ be a quasi--length-preserving scheme. We show an attack that is a generalization of the distinguishing attack in Theorem~\ref{theo:separ}.

\begin{enumerate}
\item for $m$-bit message strings, the distinguisher \D sets the two plaintext states for the qIND- game to be: $\ket{\varphi_0} = H \ket{0^m}, \ket{\varphi_1} = H \ket{1^m}$, where $H$ is the $m$-fold tensor Hadamard transformation. Notice that both these states admit efficient classical representations, and are thus allowed in the qIND game.
\item The challenger flips a random bit $b$ and returns $\ket{\psi} = U_{\Enc_k}\ket{\varphi_b}$.
\item \D applies $H$ to the core-function part of the ciphertext $\ket{\psi}$ and measures it in the computational basis. \D outputs $0$ if and only if the outcome is $0^m$, and outputs $1$ otherwise.
\end{enumerate}

As already observed, applying $U_{\Enc_k}$ to $H \ket{0^m}$ leaves the state untouched: since the encryption oracle merely performs a permutation in the basis space, and since $\ket{\varphi_0}$ is a superposition of every basis element with the same amplitude, it follows that whenever $b$ is equal to $0$, the ciphertext state will be unchanged. In this case, after applying the self-inverse transformation $H$ again, $\D$ obtains measurement outcome $0^m$ with probability 1. On the other hand, if $b=1$, $\ket{\varphi_1} = \frac{1}{2^{m/2}}\sum_y (-1)^{y \cdot 1^m} \ket{y}$ where $a \cdot b$ denotes the bitwise inner product between $a$ and $b$. Hence, $\ket{\varphi_1}$ is a superposition of every basis element where (depending on the parity of $y$) half of the elements have a positive amplitude and the other half have a negative one, but all of them will be equal in absolute value. Applying $U_{\Enc,k}$ to this state, results in $ \frac{1}{2^{m/2}}\sum_y (-1)^{y \cdot 1^m} \ket{\Enc_k(y)}$. After re-applying $H$, the amplitude of the basis state $\ket{0^m}$ becomes $\sum_y (-1)^{y \cdot 1^m + \Enc_k(y) \cdot 0^m}$ which is easily calculated to be 0. Hence, the above attack gives \D a way of perfectly distinguishing between encryptions of the two plaintext states. \qed
\end{proof}

Notice that the above attack also works if \A is allowed to send quantum states to \C directly. Therefore, it also holds for the gqIND notion of quantum indistinguishability described in Section~\ref{sec:qindqcpa}. In particular, the above theorem shows that \cite[Construction 5.3.9]{Goldreich2004}, which in~\cite{Boneh2013} was shown to be IND-qCPA if the used PRF is quantum secure, does not fulfill qIND, nor gqIND.

This attack is a consequence of the well-known fact that, in order to perfectly (information-theoretically) encrypt a single quantum bit, {\em two} bits of classical information are needed: one to hide the basis bit, and one to hide the phase (i.e.\ the signs of the amplitudes). The fact that we are restricted to quantum operations of the form  $U_{\Enc_k}$ - that is, quantum instantiations of classical encryptions - means that we cannot afford to hide the phase as well, and this restriction allows for an easy distinguishing procedure.

\subsection{Secure Construction}

Here we propose a construction of a qIND-qCPA secure symmetric-key encryption scheme from any family of quantum-secure pseudorandom permutations (see Appendix~\ref{app:defs} for formal definitions).

\begin{construction}\label{const:qIND}
For security parameter $n$, let $m= \poly{n}$ and $\tau= \poly{n}$. Consider an efficient family of permutations $\Pi_{m+\tau} = (\init, \Pi, \Pi^{-1})$ with key space $\kSpace_\Pi$ that operates on bit strings of length $m+\tau$, and consider a plaintext message space $\msgSpace = \bin^m$, key space $\kSpace = \kSpace_\Pi$, and ciphertext space $\C=\bin^{m+\tau}$. The construction is given by the following algorithms:
\begin{description}
 \item[Key generation algorithm $k \exec \Gen(1^n)$:] on input of security parameter $n$, the key generation algorithm runs $k\exec\init(1^{m+\tau})$ and returns secret key $k$.
 \item[Encryption algorithm $y \exec \Enc_k(x)$:] on input of message $x \in \msgSpace$ and key $k \in \kSpace$, the encryption algorithm samples a $\tau$-bit string $r\rand\bin^\tau$ uniformly at random, and outputs $y = \pi_k(x\|r)$ ($\|$ denotes string concatenation).
 \item[Decryption algorithm $x \exec \Dec_k(y)$:] on input of ciphertext $y \in \C$ and key $k \in \kSpace$, the decryption algorithm first runs $x' = \pi^{-1}_k(y)$, and then returns the first $m$ bits of $x'$.
\end{description}
\end{construction}

The soundness of the construction can be easily checked. The security is stated in the following theorem.

\begin{theorem}[qIND-qCPA security of Construction~\ref{const:qIND}]\label{thm:constr}
If $\Pi_{m+\tau}$ is a family of quantum-secure pseudorandom permutations (qPRP), then the encryption scheme $(\Gen,\Enc,\Dec)$ defined in Construction \ref{const:qIND} is qIND-qCPA secure.
\end{theorem}

In the next section, we prove the security of a more powerful scheme which includes the above theorem as special case of a single message block.

\subsection{Length Extension}
Construction~\ref{const:qIND} has the drawback that the message length is upper bounded by the input length of the qPRP (minus the bit length of the randomness). However, like in the case of block ciphers, we can overcome this issue with a {\em mode of operation}. More specifically, we can handle arbitrary message lengths by splitting the message into $m$-bit blocks and applying the encryption algorithm of Construction~\ref{const:qIND} independently to each message block (using the same key but new randomness for each block). This procedure is akin to a `randomized ECB mode', in the sense that each message block is processed separately, like in the ECB (Electronic Code Book) mode, but in our case the underlying cipher is inherently randomized (since we use fresh randomness for each block), so we can still achieve qCPA security. For simplicity we consider only message lengths which are multiples of $m$. The construction can be generalized to arbitrary message lengths using standard padding techniques. Moreover, the randomness for every block can be generated efficiently using a random seed and a post-quantum secure PRNG.

\begin{construction}\label{const:extension}
For security parameter $n$, let $m= \poly{n}$ and $\tau= \poly{n}$. Consider an efficient family of permutations $\Pi_{m+\tau} = (\init, \Pi, \Pi^{-1})$ with key space $\kSpace_\Pi$ that operates on bit strings of length $m+\tau$, and consider a plaintext message space $\msgSpace = \bin^{\mu m}$ for $\mu \in \NN, \mu = \poly{n}$, key space $\kSpace = \kSpace_\Pi$, and ciphertext space $\C=\bin^{\mu(m+\tau)}$. The construction is given by the following algorithms:
\begin{description}
 \item[Key generation algorithm $k \exec \Gen(1^n)$:] on input of security parameter $n$, the key generation algorithm runs $k\exec\init(1^{m+\tau})$ and returns secret key $k$.
 \item[Encryption algorithm $y \exec \Enc_k(x)$:] on input of message $x \in \msgSpace$ and key $k \in \kSpace$, the encryption algorithm splits x into $\mu$ $m$-bit blocks $x_1, \ldots, x_\mu$. For each block $x_i$, the encryption algorithm samples a new $\tau$-bit string $r_i\rand\bin^\tau$ uniformly at random, and outputs $y_i = \pi_k(x_i\|r_i)$ ($\|$ denotes string concatenation). The ciphertext is $y = y_1 \| \ldots \| y_\mu$.
 \item[Decryption algorithm $x \exec \Dec_k(y)$:] on input of ciphertext $y \in \C$ and key $k \in \kSpace$, the decryption algorithm first splits $y$ into $\mu$ $m+\tau$-bit blocks $y_1, \ldots, y_\mu$. Then, it runs $x'_i = (\pi^{-1}_k(y_i))_m$ for each block (where $(s)_m$ refers to taking the first $m$ bits of bit string $s$). It returns the plaintext $x' = x'_1, \ldots, x'_\mu$.
\end{description}
\end{construction}

The soundness of the construction can be checked easily. For the security, we observe that splitting a $\mu m$-qubit plaintext state into $\mu$ blocks of $m$-qubits
can introduce entanglement between the blocks. We will address this issue through the following technical lemma.

\begin{lemma}\label{lemma:diamond}
Let $\mathcal{E}$ be the quantum channel that takes as input an arbitrary $m$-qubit state, attaches another $\tau$ qubits in state $\ket{0}$, and then applies a permutation picked uniformly at random from $S_{2^{m+\tau}}$ to the computational basis space. Let $\mathcal{T}$ be the constant channel which maps any $m$-qubit state to the totally mixed state on $m+\tau$ qubits. Then, $\| \mathcal{E} - \mathcal{T} \|_{\diamond} \leq 2^{-\tau+2}$.
\end{lemma}

\begin{proof}
In order to consider the fact that the $m$-qubit input state might be entangled with something else, we have to start with a purification of such a state. Formally, this is a bipartite pure $2m$-qubit state $\ket{\phi}_{XY} = \sum_{x,y} \alpha_{x,y} \ket{x}_X \ket{y}_Y$ whose $m$-qubit $Y$ register is input into the channel and gets transformed into $id_X \otimes \mathcal{E}( \proj{\phi} ) = \tr_{\Pi} \proj{\psi}$ where  
$$
\ket{\psi} = \sum_{x \in \{0,1\}^m,y \in \{0,1\}^m,\pi \in S_{2^{m+\tau}}} \alpha_{x,y} \ket{x}_X \ket{\pi(y || 0)}_C \ket{\pi}_{\Pi} \, .
$$
By definition of the diamond-norm, we have to show that for any $2m$-qubit state $\rho$, we have that $\|(id \otimes \mathcal{E})(\rho) - (id \otimes \mathcal{T})(\rho)\|_{\tr} \leq 2^{-\tau+2}$. Due to the convexity of the trace distance, we may assume that $\rho = \proj{\phi}$ is pure with $\ket{\phi}_{XY} = \sum_{x,y} \alpha_{x,y} \ket{x}_X \ket{y}_Y$. Hence, we obtain
\begin{align*}
(id_X \otimes \mathcal{E})( \proj{\phi} ) &= \tr_{\Pi} \proj{\psi} \\
&=\frac{1}{2^{m+\tau}!} \sum_{x,x',y,y',\pi} \alpha_{x,y} \overline{\alpha_{x',y'}} \ketbra{x}{x'}_X \otimes \ket{\pi(y \| 0)}\bra{ \pi(y' \| 0)}_C\\
&=\frac{1}{2^{m+\tau}!} \sum_{x,x',y} \alpha_{x,y} \overline{\alpha_{x',y}} \ketbra{x}{x'}_X \otimes \sum_\pi \ket{\pi(y \| 0)}\bra{ \pi(y \| 0)}_C\\
&\quad + \frac{1}{2^{m+\tau}!} \sum_{x,x',y \neq y'} \alpha_{x,y} \overline{\alpha_{x',y'}} \ketbra{x}{x'}_X \otimes \sum_\pi \ket{\pi(y \| 0)}\bra{ \pi(y' \| 0)}_C \\
&=\sum_{x,x',y} \alpha_{x,y} \overline{\alpha_{x',y}} \ketbra{x}{x'}_X \otimes \frac{1}{2^{m+\tau}} \sum_z \ketbra{z}{z}_C\\
&\quad + \sum_{x,x',y \neq y'} \alpha_{x,y} \overline{\alpha_{x',y'}} \ketbra{x}{x'}_X \otimes \frac{1}{2^{m+\tau}(2^{m+\tau}-1)}\sum_{z \neq z'} \ketbra{z}{z'}_C\\
&=\tr_Y \proj{\phi} \otimes \tau_C + \chi_{XC}\\
&=(id_X \otimes \mathcal{T}) (\proj{\phi}) + \chi_{XC} \, ,
\end{align*}
where we defined the ``difference state'' 
$$\chi_{XC} := \sum_{x,x',y \neq y'} \alpha_{x,y} \overline{\alpha_{x',y'}} \ketbra{x}{x'}_X \otimes \frac{1}{2^{m+\tau}(2^{m+\tau}-1)}\sum_{z \neq z'} \ketbra{z}{z'}_C \, .$$

In order to conclude, it remains to show that $\| \chi_{XC} \|_{\tr} \leq 2^{-\tau+2}$. For the $C$-register $\chi_C = \frac{1}{2^{m+\tau}(2^{m+\tau}-1)}\sum_{z \neq z'} \ketbra{z}{z'}_C$, one can verify that the $2^{m+\tau}$ eigenvalues are $(c \cdot (2^{m+\tau}-1),-c,-c,\ldots,-c)$ where $c:=\frac{1}{2^{m+\tau}(2^{m+\tau}-1)}$. Hence, the trace norm (which is the sum of the absolute eigenvalues) is exactly $c \cdot 2(2^{m+\tau}-1) = 2^{-m-\tau+1}$. 

\medskip

For the $X$-register, we split $\chi_X$ into two parts $\chi_X = \xi_X - \xi'_X$ where
\begin{align*}
\xi_X &:= \sum_{x,x'} \ketbra{x}{x'} \sum_{y,y'} \alpha_{x,y} \overline{\alpha_{x',y'}} \, ,\\
\xi'_X &:= \sum_{x,x'} \ketbra{x}{x'} \sum_{y} \alpha_{x,y} \overline{\alpha_{x',y}} \, ,
\end{align*}
and use the triangle inequality for the trace norm $\|\chi_X\|_{\tr} = \|\xi_X - \xi'_X\|_{\tr} \leq \|\xi_X\|_{\tr} + \|\xi'_{X}\|_{\tr}$. Observe that $\| \xi_X \|_{\tr} = \| \sum_{x,y} \alpha_{x,y} \ket{x} \sum_{x',y'} \overline{\alpha_{x',y'}} \bra{x'} \|_{\tr} = \| \proj{s} \|_{\tr}$ for the (non-normalized) vector $\ket{s} := \sum_{x,y} \alpha_{x,y} \ket{x}$. Hence, the trace-norm $\| \xi_X \|_{\tr} = | \braket{s\mid s} | = \sum_x | \sum_y \alpha_{x,y} |^2 \leq \sum_x \sum_y |\alpha_{x,y}|^2 \cdot 2^m = 2^m$ by the Cauchy-Schwarz inequality and the normalization of the $\alpha_{x,y}$'s. Furthermore, we note that $\xi'_X$ is exactly the reduced density matrix of $\ket{\phi}_{XY}$ after tracing out the $Y$ register. Hence, $\xi'_X$ is positive semi-definite and its trace norm is equal to its trace which is 1. In summary, we have shown that
\begin{align*}
\| \chi_{XC} \|_{\tr} &= \| \chi_X \|_{\tr} \cdot \| \chi_C \|_{\tr} \leq (\|\xi_X - \xi'_X \|_{\tr} ) \cdot 2^{-m-\tau+1}\\
&\leq (\|\xi_X\|_{\tr} + \| \xi'_X \|_{\tr} ) \cdot 2^{-m-\tau+1} \leq (2^m + 1) \cdot 2^{-m-\tau+1} \leq 2^{-\tau+2} \, .
\end{align*}
\qed
\end{proof}

If we consider a slightly different encryption channel $\mathcal{E}^T$ which still maps $m$ qubits to $m+\tau$ qubits but where the permutation $\pi$ is not picked uniformly from $S_{2^{m+\tau}}$, but instead we are guaranteed that a certain set $T \subset \{0,1\}^{m+\tau}$ of outputs never occurs, we can consider such permutations w.l.o.g. as picked uniformly at random from a smaller set $S_{2^{m+\tau}-|T|}$. In this setting, we are interested in the distance of the encryption operation $\mathcal{E}^T$ from the slightly different constant channel $\mathcal{T}^T$ which maps all inputs to the $(m+\tau)$-qubit state which is completely mixed on the smaller set $\{0,1\}^{m+\tau} \setminus T$. By modifying slightly the proof of Lemma~\ref{lemma:diamond} we get the following.

\begin{corollary}\label{cor:diamond}
Let $\mathcal{E}^T$ and $\mathcal{T}^T$ be the channels defined above. Then, 
\begin{equation}\label{eq:corbound}
\| \mathcal{E}^T - \mathcal{T}^T \|_{\diamond} \leq \frac{4}{2^{\tau} - |T|/2^m} \, .
\end{equation}
\end{corollary}

We can now prove the security of Construction~\ref{const:extension}. We give the proof for gqIND-qCPA, and then qIND-qCPA follows immediately from Theorem~\ref{theo:qind-gqind}.

\begin{theorem}[gqIND-qCPA security of Construction~\ref{const:extension}]\label{thm:extension}
If $\Pi_{m+\tau}$ is a family of quantum-secure pseudorandom permutations (qPRP), then the encryption scheme $(\Gen,\Enc,\Dec)$ defined in Construction~\ref{const:extension} is gqIND-qCPA secure.
\end{theorem}

\begin{proof}

We want to show that no QPT distinguisher \D can win the gqIND-qCPA game with probability substantially better than guessing. We first transform the game through a short game-hopping sequence into an indistinguishable game for which we can bound the success probability of any such \D.

\subheading{Game 0.} This is the original gqIND-qCPA game.

\subheading{Game 1.} This is like Game 0, but instead of using a permutation drawn from the qPRP family $\Pi_{m+\tau}$, a random permutation $\pi \in S_{2^{m+\tau}}$ is chosen from the set of all permutations over $\bin^{m+\tau}$. The difference in the success probability of \D winning one or the other of these two games is negligible. Otherwise, we could use \D to distinguish a random permutation drawn from $\Pi_{m+\tau}$ from one drawn from $S_{2^{m+\tau}}$. This would contradict the assumption that $\Pi_{m+\tau}$ is a qPRP.

\subheading{Game 2.} This is like Game 1, but \D is guaranteed that the randomness used for each encryption query are $\mu$ new random $\tau$-bit strings that were not used before. In other words, the challenger keeps track of all random values used so far and excludes those when sampling a new randomness. Since in Game 1 the same randomness is sampled twice only with negligible probability, the probability of winning these two games differs by at most a negligible amount.

\subheading{Game 3.} This is like Game 2 except that the answer to each query asked by \D also contains the randomness $r_1,\ldots,r_{\mu}$ used by the challenger for answering that query. Clearly, \D's probability of winning this game is at least the probability of winning Game 2. 

When the modified gqIND game~3 starts, \D chooses two different plaintext states and sends them to the challenger, who will then choose one of them and send it back encrypted with fresh randomness $\hat{r}_1, \ldots,\hat{r}_{\mu}$. Let $Q$ denote the set of $q \cdot \mu = poly(n)$ query values used during the previous qCPA-phase. We have to consider that from this phase, \D knows a set $T \subset \{0,1\}^{m+\tau}$ of 'taken' outputs, i.e. he knows that any $\pi(x\|\hat{r}_i)$ will not take one of these values as $\hat{r}_i$ has not been used before. So, from the adversary's point of view, $\pi$ is a permutation randomly chosen from $S'$, the set of those permutations over $\bin^{m+\tau}$ that fix these $|T|$ values. In order to simplify the proof, we will consider a very conservative bound where $|T| = q \cdot \mu \cdot 2^m$, and the size of $S'$ is $|S'| = (2^{m+\tau}-|T|)!$ (notice that this bound is very conservative because it assumes that the adversary learns $2^m$ different (classical) ciphertexts for every of the $q \cdot \mu$ `taken' randomnesses, but as we will see, this knowledge will be still insufficient to win the game.)

By construction, the encryption of a $\mu m$-qubit (possibly mixed) state $\sigma$ is performed in $\mu$ separate blocks of $m$ qubits each. We are guaranteed that fresh randomness is used in each block, hence it follows from Corollary~\ref{cor:diamond} that $\Enc_k(\sigma)$ is negligibly close to the ciphertext state where the first $m+\tau$ qubits are replaced with the completely mixed state (by noting that $|T|/2^m = q \cdot \mu$ is polynomial in $n$ in our case, and hence the right-hand side of~\eqref{eq:corbound} is negligible.).
Another application of Corollary~\ref{cor:diamond} gives negligible closeness to the ciphertext state where the first $2(m+\tau)$ qubits are replaced with the completely mixed state etc. After $\mu$ applications of Corollary~\ref{cor:diamond}, we have shown that $\Enc_k(\sigma)$ is negligibly close to the totally mixed state on $\mu(m+\tau)$ qubits. As this argument can be made for any cleartext state $\sigma$, we have shown that from $\D$'s point of view, all encrypted states are negligibly close to the totally mixed state and therefore cannot be distinguished.
\qed
\end{proof}

\begin{corollary}[qIND-qCPA security of Construction~\ref{const:extension}]\label{cor:extension}
If $\Pi_{m+\tau}$ is a family of quantum-secure pseudorandom permutations (qPRP), then the encryption scheme $(\Gen,\Enc,\Dec)$ defined in Construction~\ref{const:extension} is qIND-qCPA secure.
\end{corollary}

\section{Conclusions and Further Directions}\label{sec:conclusions}

We believe that many of the current security notions used in different areas of cryptography are unsatisfying in case quantum computers become reality. In this respect, our work contributes to a better understanding of which properties are important for the long-term security of modern cryptographic primitives. Our work leads to many interesting follow-up questions.

There are many other directions to investigate, once the basic framework of `indistinguishability versus semantic security' presented in this work is completed. A natural direction is to look at quantum CCA1 security in this framework. This topic was also initiated in~\cite{Boneh2013} relative to the IND-qCPA model; it would be interesting to extend the definition of CCA1 security to stronger notions obtained by starting from our qIND-qCPA model.

In Section~\ref{sec:qIND} we left open the interesting question on whether it is possible at all to find a separating example between the notions of qIND and gqIND. That is, find a symmetric-key encryption scheme $\E$ which is qIND-secure, but not gqIND-secure. Finding such an example (or provable lack of) would shed further light on the security model we consider.

We have so far not taken into account models where the adversary is allowed to initialize the ancilla qubits used in the encryption operation used by the challenger (i.e. the $\ket{y}$ in $\ket{x,y} \mapsto \ket{x,y \xor \Enc_k(x)}$). These models lead to the study of {\em quantum fault attacks}, because they model cases where the adversary is able to `watermark' or tamper with part of the challenger's internal memory. Moreover, we have not considered superpositions of keys or randomness: these lead to a quantum study of {\em weak-key} and {\em bad-randomness} models. The authors of this paper are not aware of any results in these directions.

One outstanding open problem is to define CCA2 (adaptive chosen ciphertext attack) security in the quantum world. The problem is that in the CCA2 game the challenger has to ensure that the attacker does not ask for a decryption of the actual challenge ciphertext leading to a trivial break. While this is easily implemented in the classical world, it raises several issues in the quantum world. What does it mean for a ciphertext to be different from the challenge ciphertext? And, more importantly: \emph{How can the challenger check?} There might be several reasonable ways to solve the first issue but, as long as the queries are not classical, we are not aware of any possibility to solve the second issue without disturbing the challenge ciphertext and the query states.

Our secure construction shows how to turn block ciphers into qIND-qCPA secure schemes. An interesting research question is whether there exists a general patch transforming an IND-qCPA secure scheme into a qIND-qCPA secure one. It is also important to study how our transformation can be applied to modes of operation different from Construction~\ref{const:extension}.

\subheading{Acknowledgements.} The authors would like to thank Ronald de Wolf and Boris {\v{S}}kori{\'{c}} for helpful discussions, the anonymous reviewers for useful comments,  Marco Tomamichel for the bibstyle, and the organizers of the Dagstuhl Seminar 15371 ``Quantum Cryptanalysis'' for providing networking and useful interactions and support. T.G. was supported by the German Federal Ministry of Education and Research (BMBF) within EC SPRIDE and CROSSING. A.H. was supported by the Netherlands Organisation for Scientific Research (NWO) under grant 639.073.005 and the Commission of the European Communities through the Horizon 2020 program under project number 645622 PQCRYPTO. C.S. was supported by a 7th framework EU SIQS grant and a NWO VIDI grant. Part of this work was supported by the COST Action IC1306.

\bibliographystyle{alphaarxiv}
\bibliography{local_bib}

\appendix

\newpage

\section{Formal Definitons}\label{app:defs}
Here we give some formal definitions that we omitted in the main body as they are somewhat standard. We include them for the paper to be self-contained. We begin with detailed formal definitions for SEM-CPA and IND-CPA. Afterwards we define quantum-secure pseudorandom permutations.

\subheading{SEM-CPA and IND-CPA.} The following definitions are more precise than the ones we use in the main text. They are included here for reference and were taken from Goldreich~(\cite{Goldreich2004}).

\begin{definition}[SEM-CPA]
A secret-key encryption scheme, $(\Gen, \Enc, \Dec)$, is said to be semantically secure under chosen plaintext attacks iff for every pair of probabilistic polynomial-time oracle machines $\A_1$ and $\A_2$, there exists a pair of probabilistic polynomial-time algorithms $\A_1'$ and $\A_2'$ such that the following two conditions hold:
\begin{enumerate}
\item For every positive polynomial $p(\cdot)$, and all sufficiently large $n$ and $z \in \bin^{\poly{n}}$ it holds that
\begin{eqnarray}
&\pr\left[
\begin{array}{rl}
v = f_m(x) \qquad&{\rm where}\\
&k \exec \Gen(1^n)\\
&((S_m, h_m, f_m), \sigma ) \exec \A_1^{\Enc_k}(1^n, z)\\
&c \exec (\Enc_k(x), h_m(x)),{\rm where }\ x \exec S_m(U_{\poly{n}})\\
&v \exec \A_2^{\Enc_k}(\sigma, c)
\end{array}
\right]\nonumber\\
< &\pr\left[
\begin{array}{rl}
v = f_m(x) \qquad&{\rm where}\\
&((S_m, h_m, f_m), \sigma ) \exec \A_1'(1^n, z)\\
&x \exec S_m(U_{\poly{n}})\\
&v \exec \A_2'(\sigma, 1^{|x|},h_m(x))
\end{array}
\right] + \frac{1}{p(n)}
\end{eqnarray}
Recall that $(S_m, h_m, f_m)$ is a triplet of circuits consisting of a poly-sized circuit $S_m$ specifying a distribution over $m$-bit long plaintexts, a circuit computing an advise function $h_m:\bin^m\rightarrow\bin^*$, and a circuit computing a target function $f_m:\bin^m\rightarrow\bin^*$, and that $x$ is a sample from the distribution induced by $S_m$.
\item For every $n$ and $z$, the first elements (i.e., the $(S_m, h_m, f_m)$ part) in the random variables $\A_1'(1^n, z)$ and $A_1^{\Enc_{\Gen(1^n)}}(1^n, z)$ are identically distributed.
\end{enumerate}
\end{definition}

\begin{definition}[IND-CPA]
A secret-key encryption scheme, $(\Gen, \Enc, \Dec)$, is said to have indistinguishable encryptions under chosen plaintext attacks iff for every pair of probabilistic polynomial-time oracle machines, $\A_1$ and $\A_2$, for every positive polynomial $p(\cdot)$, and all sufficiently large $n$ and $z \in \bin^{\poly{n}}$ it holds that
$$\left|p_{n,z}^{(1)}-p_{n,z}^{(2)}\right| < \frac{1}{p(n)}$$
where
$$p_{n,z}^{(i)} \defas \pr \left[
\begin{array}{rl}
v = i &{\text where }\\
&k \exec \Gen(1^n)\\
&((x_{1},x_{2}), \sigma) \exec \A_1^{\Enc_k}(1^n, z)\\
&c \exec \Enc_k(x_i)\\
&v \exec \A_2^{\Enc_k}(\sigma, c)
\end{array}
\right]
$$
where $\left|x_1\right| = \left|x_2\right|$.
\end{definition}

Please note that there are no restrictions regarding \A's oracle queries, i.e. $\A_1$ as well as $\A_2$ are allowed to ask for encryptions of $x_1$ and $x_2$.

\subheading{Quantum PRP.} We now define quantum-secure pseudorandom permutation families. We restrict ourselves to efficient permutation families that have as domain binary strings of a certain length as these are the only ones we are using in this work. Let $S_{2^n}$ be the set of all permutations of $n$-bit strings.

\begin{definition}[Efficient Permutation Family]
Let $n \in \NN$, we call a family of permutations $\Pi_{n} = \{\pi_k:\bin^n\rightarrow\bin^n\} \subset S_{2^n}$ with key space $\kSpace_\Pi$ and domain $\bin^n$ efficient if there exists a triple of \ppttxt algorithms $(\init, \Pi, \Pi^{-1})$ such that: 
\begin{enumerate}
\item The initialization algorithm $\init(1^n)$ takes as input the parameter $n$ and outputs a random function key $k \rand \kSpace_\Pi$ from the key space. 
\item The function $\Pi$ takes as input a function key $k$ and a domain element $x$ and outputs $\pi_k(x)$. 
\item The function $\Pi^{-1}$ takes as input a function key $k$ and a domain element $x$ and outputs $\pi^{-1}_k(x)$.
\end{enumerate}
\end{definition}

We sometimes abuse notation and write $\pi$ instead of $\pi_k$ and $\pi \rand \Pi_{n}$ for the process of running $\init(1^n)$. A quantum-secure pseudorandom permutation family (qPRP) is an efficient permutation family that achieves the pseudorandomness property in presence of a quantum adversary that can query the permutation $\pi$ with superpositions of domain elements $x$. It is defined as follows:

\begin{definition}[Quantum PRP]
An efficient permutation family $\Pi_{n}$ is said to be a \emph{quantum-secure pseudorandom permutation family} if for every quantum polynomial-time oracle machine $\A$, it holds that
$$\left| \pr_{\pi\rand\Pi_{n}}\left[\A^{\ket{\pi}}(1^n) = 1\right] - \pr_{\pi\rand S_{2^n}} \left[\A^{\ket{\pi}}(1^n) = 1\right]\right| \leq \negl{n} \, ,$$
where the superscript $\ket{\cdot}$ denotes oracle access in superposition.
\end{definition}

Note that the permutations are chosen by the game. Hence, keys are classical. 

A permutation family $\Pi_n$ is called a \emph{strong quantum PRP}, if a random member of $\Pi_n$ is computationally indistinguishable from a uniform permutation even if the attacker $\A$ can query (in superposition) both the permutation $\pi$ and the inverse permutation $\pi^{-1}$. Notice that the construction in Theorem~\ref{thm:constr} does not require {\em strong} quantum PRPs. The reason is that, even if we are considering type-$(2)$ transformations (which could be used to compute $\pi^{-1}$), these transformations are implemented by the challenger, because we are in the $(\C)$ model. And since we only consider CPA scenarios here, and not CCA, the adversary is never granted access to the decryption oracle. Hence, $\pi^{-1}$ is not needed by the reduction.

\section{Example Encryption Scheme}\label{app:examples}

In this section we recall Construction 5.3.9 from \cite{Goldreich2004} which achieves IND-CPA security starting from a pseudorandom function family.

\begin{construction}[{\cite[Construction 5.3.9]{Goldreich2004}}]\label{const:goldreich}
Let $n\in\NN$ be the security parameter, $\tau,m \in \poly{n}$, $\ff = \left\{\F_k:\bin^\tau\rightarrow\bin^m \mid k \in \kSpace\right\}$ be a pseudorandom function family with key space \kSpace. Then the following triple of algorithms form a symmetric-key encryption scheme with message space $\bin^m$:
\begin{description}
 \item[\Gen($1^n$):] On input of the security parameter, returns a uniformly random key $k \rand \kSpace$ for the PRF $\ff$ as secret key.
 \item[\Enc($x,k$):] On input of message $x$ and key $k$ returns cipher text $c = (r,c')$ where randomness $r \rand \bin^\tau$ is a uniformly random $\tau$ bit string and $c'$ is computed as 
 $$c' \exec \F_k(r) \xor x.$$
 \item[\Dec($c,k$):] On input of cipher text $c = (r,c')$ and key $k$ returns plain text
 $$x \exec c' \oplus \F_k(r).$$ 
\end{description}
\end{construction}

\section{Proof of Theorem~\ref{theo:gqind-qindcpa2}}\label{app:BJ15equiv}

In this section we explain how the q-IND-CPA-2 indistinguishability notion for secret-key encryption of quantum messages introduced by Broadbent and Jeffery~\cite[Appendix~B]{BJ15} is equivalent to our gqIND-qCPA notion in the case that the encryption operation is a symmetric-key classical functionality operating in type-(2) mode. In~\cite{BJ15}, the authors study the definition of quantum indistinguishability relative to the case of {\em quantum fully homomorphic encryption}. The general definition of {\em quantum symmetric-key encryption scheme} has been formalized in~\cite{Alagic+16arxiv} in the following way.

\begin{definition}\label{def:QSKE}
A {\em quantum symmetric-key encryption scheme (or qSKE)} is a triple of quantum circuit families of polynomial depth:
\begin{enumerate}
\item (key generation) $Q.\Gen: 1^n \mapsto k \in \K$
\item (encryption) $Q.\Enc: \K \times \X \rightarrow \Y$
\item (decryption) $Q.\Dec: \K \times \Y \rightarrow \X$
\end{enumerate}
such that $\| Q.\Dec \circ Q.\Enc - \one_\X \|_\diamond \leq \negl{n}$
for all $k \in \supp{Q.\Gen(1^n)}$, where $\K$ is the {\em key space}, $\X$ is the {\em plainstate space}, $\Y$ is the {\em cipherstate space}, $\one$ is the identity operator, and $Q.\Dec,Q.\Enc$ must be intended acting with the same (classical) key $k$.
\end{definition}

Then the authors of~\cite{Alagic+16arxiv} define a notion of {\em quantum indistinguishability for quantum symmetric-key encryption schemes} (which they call IND, but which we relabel here as q-IND-qse for ease of reading) as follows.

\begin{definition}[q-IND-qse]\label{def:qINDqse}
A qSKE $(Q.\Gen, Q.\Enc, Q.\Dec)$ has {\em indistinguishable encryptions} (or is {\em q-IND-qse secure}) if for every QPT adversary $\A=(\M,\D)$ we have:
\begin{equation*}
\left| \Pr \left[ \;\D_{Q.\Enc}  (\rho_{ME})  = 1 \; \right] -
\Pr \left[ \; \D_{Q.\Enc}  (\ketbra{0}{0}_M \otimes \rho_E)  = 1 \;  \right] \right| \leq \negl{n}
\end{equation*}
where $\rho_{ME} \from \M$, $\rho_E = \tr_M(\rho_{ME})$, $\D_{Q.\Enc} = \D \circ (\Enc_{k} \otimes \one_E)$ and the probabilities are taken over $k \leftarrow Q.\Gen(1^n)$ and the internal randomness of \Enc, $\M$, and $\D$.
\end{definition}

Basically, the above definition states that for any QPT adversary $\A$, it must be hard to distinguish an encryption of any state $\rho_M$ from an encryption of $\ketbra{0}{0}_M$ (where $\rho_E$ is auxiliary information carried between the two parts $\M$ and $\D$ of $\A$). Once we add a quantum CPA phase ($\M$ and $\D$ are given oracle access to $\Enc_k$), Definition~\ref{def:qINDqse} translates to the notion of q-IND-CPA from~\cite{BJ15}. And, also in~\cite[Theorem~B.2]{BJ15}, this notion q-IND-CPA has been shown to be equivalent to another notion, q-IND-CPA-2, which considers the case where in the above game there are two messages chosen by the adversary, $\rho^0$ and $\rho^1$, instead of a single state $\rho$ and the fixed $\ketbra{0}{0}$ state. In other words, the q-IND-CPA-2 game can then be summarized as follows.

\begin{definition}[q-IND-CPA-2]\label{def:qINDCPA2}
A qSKE $(Q.\Gen, Q.\Enc, Q.\Dec)$ is {\em q-IND-CPA-2 secure}) if any QPT adversary $\A$ having oracle access to $Q.\Enc_{k}$ has probability at most negligibly better than guessing of winning the following game:
\begin{enumerate}
\item \A generates two plaintext state messages $\rho^0,\rho^1 \in \X$ and sends them to the challenger \C;
\item \C flips a random bit $b \rand \bin$;
\item \C traces out (discards) $\rho^{1-b}$;
\item \C encrypts $\rho^b$ to $\varphi \from Q.\Enc_k (\rho^b)$;
\item \A receives back $\varphi$ from \C;
\item \A outputs a bit $b'$, and wins the game iff $b=b'$.
\end{enumerate}

\end{definition}

Finally, notice that Definition~\ref{def:qINDCPA2} is equivalent to Definition~\ref{def:gqIND} when the encryption algorithm $Q.\Enc$ is actually a type-(2) unitary operator $U_\Enc$ of a classical simmetric-key encryption scheme $(\Gen,\Enc,\Dec)$. This concludes the proof of Theorem~\ref{theo:qind-gqind}.\qed

\section{Semantic Security with Quantum Advice States}\label{app:qaSEM}

In Section~\ref{qasemquestion} we left open the question of what happens if the messages (and the function to be computed about the message) are still classical, but the auxiliary advice can be a quantum state. Here we discuss this scenario.

A possible first approach is the following: Let $U_{\xi_m}$ be a unitary (the {\em advice unitary}) that takes as input a basis element $\ket{x}$ representing a classical $m$-bit message $x$ as well as (if required) an auxiliary register prepared by \C and computes a quantum advice state $\ket{\xi_m}$. Then we can define the following challenge phase and the corresponding notion.

\subheading{Quantum-advice SEM challenge phase (qaSEM):} \A sends \C a challenge template consisting of: a poly-sized classical circuit $S_m$ specifying a distribution over $m$-bit plaintexts $x$, a classical description of the advice unitary $U_{\xi_m}$, and a target function $f_m:\bin^m \rightarrow \bin^{\poly{n}}$ for an $m\in\NN$ of \A's choice. \C replies with the pair $(\Enc_k(x), \ket{\xi_m})$, where $x$ is sampled according to $S_m$ and $\ket{\xi_m}$ is computed by constructing and evaluating $U_{\xi_m}$ on $\ket{x}$. \A's goal is to output $f_m(x)$. Again, \sim plays in the reduced game and learns only $\ket{\xi_m}$.

\begin{definition}[qaSEM-qCPA]\label{def:qasemqcpa} A secret-key encryption scheme is said to be qaSEM-qCPA-secure if for every quantum polynomial-time machine \A, there exists a quantum polynomial-time machine \sim such that the challenge templates produced by \sim and \A are identically distributed and the success probability of \A winning the qaSEM-qCPA game is negligibly close (in $n$) to the success probability of \sim winning the reduced game. 
\end{definition}

At a first glance it might seem as if qaSEM-qCPA is equivalent to SEM-qCPA as a security notion because having a classical advice function $h(x)$ is just a special case of a quantum advice circuit depending on $x$. Notice however that as we restrict $U_{\xi_m}$ to be a circuit computing a unitary operator $U\ket{x}$ this notion is meaningless because it is trivially achievable by {\em any} encryption scheme. The reason is that, in this case, both \A and \sim can always apply $U^{-1}$ to $\ket{\xi_m}$ to recover the message -- it is like restricting the classical notion to the case where the advice function $h$ is just a permutation {\em chosen} by \A (resp. \sim). 

To fix this problem, we have to allow more general quantum circuits $U'_{\xi_m}$ that can somehow provide non-reversible information, for example by applying some partial measurement at the end, or by providing \A (resp. \sim) only with {\em some} output qubits, while \C keeps the others. Towards this end let $U'_{\xi_m}$ be an arbitrary quantum circuit (the {\em advice circuit}) that takes as input a basis element $\ket{x}$ representing a classical $m$-bit message $x$, a quantum state $\rho_m$ provided by \A (resp. \sim) (that includes possibly needed auxiliary registers), and computes a quantum advice state $\xi_m$. This leads to the following definition:

\subheading{Ideal quantum advice, classical SEM challenge phase (iqSEM):} \A sends \C a challenge template consisting of: a poly-sized classical circuit $S_m$ specifying a distribution over \mbox{$m$-bit} plaintexts, a classical description of the quantum advice circuit $U'_{\xi_m}$, a quantum state $\rho_m$, and a target function $f_m:\bin^m \rightarrow \bin^{\poly{n}}$ for an $m\in\NN$ of \A's choice. \C replies with the pair $(\Enc_k(x), \xi_m)$, where $x$ is sampled according to $S_m$ and $\xi_m$ is computed by constructing and executing $U'_{\xi_m}$. \A's goal is to output $f_m(x)$.

The iqSEM-qCPA game is defined by qCPA learning phases and a iqSEM challenge phase. This leads to the following definition:

\begin{definition}[iqSEM-qCPA]\label{def:iqasemqcpa} A secret-key encryption scheme is said to be iqSEM-qCPA-secure if for every quantum polynomial-time machine \A, there exists a quantum polynomial-time machine \sim such that the challenge templates produced by \sim and \A are identically distributed and the success probability of \A winning the iqSEM-qCPA game is negligibly close (in $n$) to the success probability of \sim winning the reduced game. 
\end{definition}

This notion turns out to be equivalent to SEM-qCPA (and IND-qCPA). The reason is that having a quantum advice state does not really give any additional power to \A in the case of classical messages and target functions. This can be seen from the reduction between IND-qCPA and SEM-qCPA -- see the proofs of Propositions~\ref{prop:indq:semq} and~\ref{prop:semq:indq}. In one case, the advice state is only used to pass \A's code from the first circuit of \sim to the second one (which can also be done with a quantum advice state), in the other case it is set to a constant function. 

It seems like introducing arbitrary quantum advice circuits (as opposed to {\em superpositions of classical advices}) is not meaningful as long as the messages are still classical. Consequently, we proceed in Section~\ref{sec:qsem} with our search for a notion of quantum semantic security considering quantum message states.

\end{document}